\documentclass[11pt,a4paper]{article}

\usepackage{titlesec}
\usepackage{amssymb}
\usepackage{amsmath,amsthm,amssymb,color}
\usepackage[pdftex,pagebackref,colorlinks]{hyperref}
\usepackage{graphicx}
\usepackage{multirow}
\usepackage{makecell}
\usepackage{booktabs}
\usepackage{enumerate}
\date{}
\usepackage{url}
\urldef{\mailsa}\path|{alfred.hofmann, ursula.barth, ingrid.haas, frank.holzwarth,|
\urldef{\mailsb}\path|anna.kramer, leonie.kunz, christine.reiss, nicole.sator,|
\urldef{\mailsc}\path|erika.siebert-cole, peter.strasser, lncs}@springer.com|

\newtheorem{theorem}{Theorem}
\newtheorem{lemma}{Lemma}
\newtheorem{corollary}{Corollary}
\newtheorem{definition}{Definition}
\newtheorem{proposition}{Proposition}

\newtheorem{remark}{Remark}

\newtheorem{construction}{Construction}

\begin{document}

\title{Minimal Ternary Linear Codes from Vectorial Functions}

\author{Yanjun Li, Haibin Kan, Fangfang Liu, Jie Peng,  Lijing Zheng, and Zepeng Zhuo\thanks{Corresponding author
\newline\indent Y. Li is with Institute of Statistics and Applied Mathematics, Anhui University of Finance and Economics,  Bengbu, Anhui 233030,   China; and School of Computer Sciences, Fudan University, Shanghai 200433, China
(yanjlmath90@163.com).
\newline \indent H. Kan is with Shanghai Key Laboratory of Intelligent Information Processing, School of Computer Science, Fudan University; Shanghai Engineering Research Center of Blockchain, Shanghai Institute for Advanced Communication and Data Science, Shanghai 200433, China (Emails:~hbkan@fudan.edu.cn).
\newline \indent F. Liu and Z. Zhuo  are with School of Mathematics and Statistics,
Huaibei Normal University, Huaibei, Anhui 235000, China (E-mails:~19810900401@163.com and zzp781021@sohu.com).
\newline \indent J. Peng is with Mathematics and Science College of
Shanghai Normal University, Guilin Road \#100, Shanghai 200234, China (E-mail:~jpeng@shnu.edu.cn).
\newline \indent L. Zheng is  with the School of Mathematics and Physics, University of South China, Hengyang, Hunan 421001, China (Email:~zhenglijing817@163.com).}}
\maketitle
\begin{abstract}
The study on minimal linear codes has received great attention due to their significant applications in secret sharing schemes and secure two-party computation. Until now, numerous minimal linear codes have been discovered.
However, to the best of our knowledge, no infinite family of minimal ternary linear codes was found from vectorial functions. In this paper, we present a necessary and sufficient condition for a large class of ternary linear codes from vectorial functions such that those codes are minimal. Based on that, we construct several minimal ternary linear codes with three-weight from vectorial regular plateaued functions, and determine their weight distributions. Moreover, we also give a necessary and sufficient condition for a large family of ternary linear codes from vectorial functions such that the codes are minimal and violate the AB condition simultaneously. According to this characterization, we find several minimal ternary linear codes violating the AB condition. Notably, our results show that our method can be applied to solve a problem on minimal linear codes proposed by  Li et al.
\end{abstract}



\section{Introduction}

Linear codes with few weights are important in  communication, consumer electronics and data storage system.
 As a special kind of linear codes, minimal linear codes have attracted  enormous research attention because of their wide and significant applications in  secret sharing schemes \cite{Carlet-ding-2007-IT, Yuan-2006-IT} and secure two-party computation \cite{Chabanne-2014, Ding-2003}.

In 1998, according to the ratio of the maximum and minimum
nonzero Hamming weights in a linear code $C$, Ashikhmin and Barg \cite{AB-condition-1998} found a sufficient condition (which is called the AB condition in this paper) to judge whether $C$ is minimal (see Lemma \ref{lemABcondition} of this paper). A number of linear codes satisfying the AB condition have been reported from the study of linear codes with few weights (see, e.g., \cite{Ding-2003, Ding-2015-IT2, Hengzl-2016-FFA, Luogj-2018, Sihem-2017, Sihem-2019, Shimj-2017-IT, Tangcm-2016-IT, Yang-2017-DCC, Zhouzc-2016-DCC}). In \cite{Cohen-2003}, Cohen et al. gave the first sporadic  example to show  that the AB condition is not necessary for a linear code to be minimal.
However, there is no infinite family of minimal linear codes violating the AB condition until Chang and
Hyun's breakthrough in \cite{Chang-2018-DCC}, where the first infinite family of such codes was found from the linear code
\begin{align}\label{eqcodeBooleanintro}
C_f=\bigg\{c(\mu,\nu)=\big(\mu f(x)+\nu \cdot x\big)_{x\in\mathbb{F}_{p}^{n*}} : \mu\in\mathbb{F}_p,\nu\in\mathbb{F}_{p}^{n}\bigg\},
\end{align}
where $f$ is a function from $\mathbb{F}_p^n$ to $\mathbb{F}_p$.
 Then for the minimality of the linear code $C_f$ given in \eqref{eqcodeBooleanintro}, when $p=2$, by using the Walsh transform of Boolean functions, Ding et al. \cite{Ding-2018-IT} gave a complete characterization (that is, they found a necessary and sufficient condition such that $C_f$  is minimal). Based on that characterization, Ding et al. found three infinite families of minimal binary linear codes violating the AB condition. In \cite{Heng-2018-FFA}, the authors studied the minimality of $C_f$ for $p=3$, and presented an infinite family of minimal ternary linear codes violating the AB condition. In \cite{Bartoli-2019-IT}, the authors studied the minimality of $C_f$ for any odd prime $p$, and provided infinite families of minimal linear codes violating the AB condition.

Besides, under the construction of linear code $C_f$ defined as in \eqref{eqcodeBooleanintro}, many other minimal linear codes violating the AB condition were also discovered by selecting  different kinds of functions $f$ from $\mathbb{F}_p^n$ to $\mathbb{F}_p$. For instance, such linear codes were found from the functions $f$ with $V(f)=\{x\in\mathbb{F}_p^n:f(x)=0\}$ being a cutting blocking set \cite{Bonini-2020-JAC, Pasalic-2021-CCDS}; from  characteristic functions \cite{Sihem-2020-IT, taoran-2021-IT}; from Maiorana-McFarland functions \cite{Duxiaoni-2024-DCC, Rodriguez-2023-IT, Xu-2020-FFA, Zhangfengrong-2021-DCC}; from weakly regular plateaued/bent functions \cite{Sihem-2020-IT2, Xu-2021-CCDS}; and from the direct sum functions \cite{Rodriguez-2023-IT, Zhangfengrong-2022-DCC}.

In \cite{LYJetal-2023-IT}, the authors generalized the construction of linear codes $C_f$ given in \eqref{eqcodeBooleanintro} to the following construction:
\begin{align}\label{lietalcon}
C_F=\bigg\{c(\mu,\nu)=\big(\mu \cdot F(x)+\nu \cdot x\big)_{x\in\mathbb{F}_{p}^{n*}} : \mu\in\mathbb{F}_p^m,\nu\in\mathbb{F}_{p}^{n}\bigg\},
\end{align}
where $F$ is a function from $\mathbb{F}_p^n$ to $\mathbb{F}_p^m$, $m\geq2$. Then when $p=2$, by using the Walsh transform of vectorial Boolean functions,  they gave a necessary and sufficient condition such that $C_F$ is minimal and violates the AB condition. Based on that, they also obtained several infinite families of minimal binary linear codes with higher dimensions and  violating the AB condition. By choosing some two-to-one functions $F$ in Construction \eqref{lietalcon}, Mesnager et al. \cite{Sihem-2023-IT} also presented several infinite families of minimal binary linear codes with few weights and  flexible parameters. However, to the best of our knowledge, when $p\neq2$, no infinite family of minimal linear codes was found from Construction \eqref{lietalcon}.

In this paper, we study the case of $p=3$ for Construction \eqref{lietalcon}. We first completely determine the parameters (including the length, dimension and minimal distance) of the ternary linear codes $C_F$, and then present a necessary and sufficient condition such that $C_F$ is minimal by using the Walsh transform of the functions $F:\mathbb{F}_3^n\to\mathbb{F}_3^m$. Based on that, we find several minimal linear codes with three-weight and completely determine their weight distributions. Moreover, we also give a necessary and sufficient condition such that the ternary linear code $C_F$ is minimal and violates the AB condition simultaneously, by analyzing the properties of the function $F=(f,G)$, where $f$ is a function from $\mathbb{F}_3^n$ to $\mathbb{F}_3$ and $G$ is a function from $\mathbb{F}_3^n$ to $\mathbb{F}_3^{m-1}$.
According to this condition, we construct several  minimal ternary linear codes violating the AB condition. We emphasize that the minimal ternary codes violating the AB condition obtained in this paper are valid for any positive integers $n\geq 6$ and $m\geq2$ (without the restrictions for $n$ even and $2\leq m\leq \frac{n}{2}$ or $m=n+1$, which are restricted in \cite{LYJetal-2023-IT}), and hence one can use our method to  answer the question given in Remark 12 of \cite{LYJetal-2023-IT}.

The rest of this paper is organized as follows. In Section \ref{sec:Preliminaries}, we fix some notations, recall some  essential definitions, and introduce some necessary knowledge related to functions from $\mathbb{F}_p^n$ to $\mathbb{F}_p$ and linear codes. In Section \ref{sec:generic1}, we fully characterize the parameters and minimality of a large class of ternary linear codes. In Section \ref{sec:threeweight}, we present several minimal ternary linear codes with three-weight and completely determine their weight distributions. Section \ref{sec:mvab} gives a generic construction of minimal ternary linear codes violating the AB condition. Based on that, several minimal ternary linear codes violating the AB condition are constructed in Section \ref{sec:appli}. Finally, Section \ref{sec:conclusion} summarizes the main contributions of this paper.

\section{Preliminaries}\label{sec:Preliminaries}

 In this paper, $n$ and $m$ are two positive integers. For any set $E$, let $E^{*}=E \setminus\{0\}$ and let $\# E$ be the cardinality of $E$. For a complex number $\xi\in \mathbb{C}$, $|\xi|$ denotes the (complex) magnitude of $\xi$.
Let $p$ be a prime number, $\mathbb{F}_{p^n}$ be the finite field of $p^n$ elements, and
$\mathbb{F}_p^n$ be the $n$-dimensional linear space over $\mathbb{F}_p$. 


\subsection{Functions and Walsh transform over $\mathbb{F}_p^n$ }

Let $f$ be a function from $\mathbb{F}_p^n$ to $\mathbb{F}_p$. Then $f$ can be uniquely determined by the following polynomial (called the algebraic normal form (ANF) of $f$):
\begin{align*}
f(x_1,x_2,\ldots,x_n)=\sum_{(j_1,j_2,\ldots,j_n)\in \mathbb{F}_p^n}a_{(j_1,j_2,\ldots,j_n)}\prod_{i=1}^nx_i^{j_i}, \hspace{0.5cm}a_{(j_1,j_2,\ldots,j_n)}\in\mathbb{F}_p.
\end{align*}
The algebraic degree of $f$, denoted by $\deg(f)$, is defined as the maximum value $j_1+j_2+\cdots+j_n$ in the ANF of $f$ such that $a_{(j_1,j_2,\ldots,j_n)}\neq 0$. 
 If $\deg(f)\leq 1$, then $f$ is called affine.

The {\it direct and inverse Walsh transform} of a function $f:\mathbb{F}_p^n\to \mathbb{F}_p$ at the point $\nu\in\mathbb{F}_p^n$ are defined, respectively, by
\begin{align*}
W_f(\nu)=\sum_{x\in \mathbb{F}_p^n}\zeta_p^{f(x)-\nu\cdot x} \hspace{0.3cm}{\rm and} \hspace{0.3cm}
\zeta^{f(\nu)}=\frac{1}{p^n}\sum_{x\in \mathbb{F}_p^n}W_f(x)\zeta_p^{\nu\cdot x},
\end{align*}
 where $\zeta_p=e^{2\pi \sqrt{-1}/p}$ denotes a primitive $p$-th complex root of unity, $\nu\cdot x$ is the standard inner product of $\nu$ and $x$ in $\mathbb{F}_p^n$. 
  The inverse Walsh transform shows that any function $f:\mathbb{F}_{p^n}\to\mathbb{F}_p$ can be recovered from its Walsh transforms. 

 By using the fact that $|\xi|^2=\xi \bar{\xi}$ for any complex number $\xi$, and
 \begin{align}\label{eqlinearwalsh}
 \sum_{x\in\mathbb{F}_p^n}\zeta_p^{\nu\cdot x}
 =\begin{cases}
 p^n,~~&{\rm if}~\nu={\bf 0_n},\\
 0,~~&{\rm otherwise},
 \end{cases}
 \end{align}
 where ${\bf 0_n}$ is the zero vector of $\mathbb{F}_p^n$, we can deduce the following result.

\begin{lemma}\label{Titsworth}
Let $f$ be a function from $\mathbb{F}_p^n$ to $\mathbb{F}_p$. Then it holds that
\begin{align*}
\sum_{\nu\in\mathbb{F}_p^n}W_f(\nu)\overline{W_f(\nu+\tau)}
=\begin{cases} p^{2n}, ~~&{\rm if}~\tau={\bf 0_n},\\
 0,~~&{\rm otherwise}.
\end{cases}
\end{align*}
\end{lemma}
Lemma \ref{Titsworth} is the so-called {\it Titsworth's theorem} when $p=2$. The case $\tau={\bf 0_n}$ of Lemma \ref{Titsworth}, that is,
$\sum_{\nu\in\mathbb{F}_p^n}|W_f(\nu)|^2=p^{2n}$,
 is the well-known {\it Parseval's identity}. It implies that $\max_{\nu\in\mathbb{F}_p^n}|W_f(\nu)|\geq p^{n/2}$. The equality holds if and only if $|W_f(\nu)|=p^{n/2}$ for all $\nu\in\mathbb{F}_p^n$. In this case, $f$ is called {\it bent}.

Note that when $p=2$, one has $\overline{W_f(\nu+\tau)}=W_f(\nu+\tau)$. Thus, by Lemma \ref{Titsworth}, $f$ is affine if $W_f(\nu)\leq 0$ (or $\geq0$) for any $\nu\in\mathbb{F}_p^n$ (see also Lemma 1 of \cite{LYJetal-2023-IT}). When $p> 2$, we conjecture that $f$ is affine if $Re(W_f(\nu))\leq 0$ (or $\geq0$) for any $\nu\in\mathbb{F}_p^n$, where $Re(W_f(\nu))$ denotes the real part of the complex number $W_f(\nu)$. But we cannot prove this result, and so we invite the interested readers to attack.

\begin{definition}\label{defs-plateaued}
Let $s$  be an integer with $0\leq s\leq n$. A function $f$ from $\mathbb{F}_p^n$ to $\mathbb{F}_p$ is called {\it $s$-plateaued} if $|W_f(\nu)|^2\in\{0,p^{n+s}\}$ for all $\nu\in\mathbb{F}_p^n$.
\end{definition}
According to Definition \ref{defs-plateaued}, it is easily seen that any bent function is $0$-plateaued.  

 In \cite[Theorem 2]{Hyun-2016-IT}, Mesnager et al. \cite{Sihem-2019} introduced the following definition of {\it weakly regular} plateaued functions.

\begin{definition}\label{defregular-pla}
Let $f:\mathbb{F}_p^n\to\mathbb{F}_p$ be an $s$-plateaued function. If there is a complex number $\epsilon\in\{\pm1,\pm\sqrt{-1}\}$ such that $W_f(\nu)\in\{0,\epsilon p^{\frac{n+s} {2}} \zeta_p^{g(\nu)}\}$ for all $\nu\in\mathbb{F}_p^n$, then $f$ is called {\it weakly regular $s$-plateaued}, where $g$ is a function from $\mathbb{F}_p^n$ to $\mathbb{F}_p$ with $g(x)=0$ for any $x\in\mathbb{F}_p^n\backslash S_f:=\{\nu\in\mathbb{F}_p^n:W_f(\nu)\neq 0\}$. In particular, $f$ is said to be {\it regular} if $\epsilon=1$.
\end{definition}


 A mapping $F$ from $\mathbb{F}_p^n$ to $\mathbb{F}_p^m$ is called an {\it$(n,m)$-function}.  The function $\mu\cdot F$ is  called the {\it component} of $F$ in terms of $\mu\in\mathbb{F}_p^m$. The {\it algebraic degree} of $F$ is defined to be the largest algebraic degree of its components.
 The {\it Walsh transform} of $F$ at the point $(\mu,\nu)\in\mathbb{F}_p^m\times \mathbb{F}_p^n$ is defined by the Walsh transform of $\mu\cdot F$ at the point $\nu\in\mathbb{F}_p^n$, that is,
\begin{align*}
W_F(\mu,\nu)=\sum_{x\in\mathbb{F}_p^n}\zeta_p^{\mu\cdot F(x)-\nu\cdot x}.
\end{align*}

\begin{definition}\label{defvecpla}
An $(n,m)$-function $F$ is said to be vectorial plateaued if $\mu\cdot F$ is plateaued for all $\mu\in\mathbb{F}_p^{m*}$, and $F$ is said to be vectorial $s$-plateaued if $\mu\cdot F$ is $s$-plateaued for all $\mu\in\mathbb{F}_p^{m*}$. In particular, $F$ is called vectorial (weakly) regular $s$-plateaued if $\mu\cdot F$ is (weakly) regular $s$-plateaued for all $\mu\in\mathbb{F}_p^{m*}$.
\end{definition}

When $s=0$, vectorial $0$-plateaued functions $F:\mathbb{F}_p^n\to\mathbb{F}_p^m$ are also called {\it vectorial bent} functions, which exist for $m\leq n$, and $F$ is called a {\it planar} function if $n=m$. It is shown in \cite{Nyberg} that vectorial regular bent functions only exist for $m\leq\frac{n}{2}$.

\subsection{Minimal linear codes}

Let $c=(c_1,c_2,\ldots,c_n)\in\mathbb{F}_p^n$. The set ${\rm suppt}(c):=\{1\leq i\leq n:c_i\in\mathbb{F}_p^*\}$ is called the {\it support} of $c$, whose cardinality is called the {\it (Hamming) weight} of $c$, denoted by ${\rm wt}(c)$, i.e., ${\rm wt}(c)=\#\{1\leq i\leq n:c_i\in \mathbb{F}_p^*\}$. A vector $c^{(1)}\in\mathbb{F}_p^n$ is {\it covered} by a vector $c^{(2)}\in\mathbb{F}_p^n$ if ${\rm suppt}(c^{(1)})\subseteq{\rm suppt}(c^{(2)})$. We denote it by $c^{(1)}\preceq c^{(2)}$. 


A code $C\subseteq \mathbb{F}_p^n$ is called an $[n,k,d]$ linear code if it is a $k$-dimensional linear subspace of $\mathbb{F}_p^n$ with minimum nonzero Hamming weight $d$, where $n$, $k$ and $d$ are called the {\it length, dimension and minimal distance} of $C$, respectively. The {\it weight enumerator} of $C$ is defined by
$1+A_1z+A_2z^2+\cdots+A_nz^n$, where $A_i$ is the number of codewords in $C$ whose weight is $i$. A code $C$ is said to be a {\it$t$-weight} code if the number of nonzero $A_i$ in the weight enumerator of $C$ is equal to $t$.

 A codeword $c\in C$ is called {\it minimal} if it only covers $\lambda c$ for any $\lambda\in\mathbb{F}_p$. A linear code $C$ is said to be minimal if every codeword in $C$ is minimal. The following result is a sufficient condition for a linear code to be minimal.

\begin{lemma}\cite{AB-condition-1998}\label{lemABcondition}
Let $C\subseteq \mathbb{F}_p^n$ be a linear code. Let $w_{\max}$ and $w_{\min}$ be the maximum and minimum  nonzero  Hamming weights in $C$, respectively. If
\begin{align}\label{eqAB-condition}
\frac{w_{\min}}{w_{\max}}> \frac{p-1}{p},
\end{align}
then $C$ is minimal.
\end{lemma}

Condition \eqref{eqAB-condition} is called the {\it AB condition} in this paper, which is only a sufficient (but not necessary) condition for a linear code to be minimal.  The following result provides a necessary and sufficient condition for a linear code to be minimal.

\begin{theorem}\cite{Heng-2018-FFA}\label{thhengFFA}
Let $C\subseteq \mathbb{F}_p^n$ be a linear code. Then $C$ is minimal if and only if
 \begin{align}\label{eqthheng}
 \sum_{\lambda\in\mathbb{F}_p^{n*}}{\rm wt}(\lambda c^{(1)}+c^{(2)})\neq (p-1){\rm wt}(c^{(2)})-{\rm wt}(c^{(1)}),
 \end{align}
 for any $\mathbb{F}_p$-linear independent codewords $c^{(1)},c^{(2)}\in C$.
\end{theorem}

\begin{remark}\label{rmkq=3}
 When $p=3$, Theorem \ref{thhengFFA} is reduced to that $C\subseteq \mathbb{F}_3^n$ is minimal if and only if
\begin{align}\label{eqhengq=3}
 {\rm wt}(c^{(1)}+c^{(2)})+{\rm wt}(-c^{(1)}+c^{(2)})\neq 2{\rm wt}(c^{(2)})-{\rm wt}(c^{(1)}),
 \end{align}
 for any $c^{(1)},c^{(2)}, c^{(1)}+c^{(2)},c^{(1)}-c^{(2)}\in C^*$.
\end{remark}

\section{A generic construction of minimal ternary  linear codes}\label{sec:generic1}

In this section, we give a generic construction of ternary linear codes from vectorial functions, and determine its parameters and minimality. 

\begin{construction}\label{ourcon}
Let $F$ be a vectorial function from $\mathbb{F}_3^n$ to $\mathbb{F}_3^m$. Define a linear code as follows:
\begin{align*}
C_F=\bigg\{c(\mu,\nu)=\big(\mu\cdot F(x)+\nu\cdot x\big)_{x\in\mathbb{F}_{3}^{n*}} : \mu\in\mathbb{F}_3^m,\nu\in\mathbb{F}_3^n\bigg\}.
\end{align*}
\end{construction}

\subsection{The parameters of $C_F$}

We first determine the length and dimension of $C_F$ in the following proposition.

\begin{proposition}\label{propdim}
Let $F$ be a function from $\mathbb{F}_3^n$ to $\mathbb{F}_3^m$ such that $\mu\cdot F(x)$ is not linear for any $\mu\in \mathbb{F}_3^{m*}$. Then the linear code $C_F$ generated by Construction \ref{ourcon} has length $3^n-1$ and dimension $n+m$.
\end{proposition}
\begin{proof}
The proof is similar as that of \cite[Proposition 1]{LYJetal-2023-IT}, we omit it here.
\end{proof}

We now characterize the weight distribution of $C_F$.

\begin{proposition}\label{weightdistri}
Let $F$ be a function from $\mathbb{F}_{3}^n$ to $\mathbb{F}_{3}^m$ with $F({\bf 0_n})={\bf 0_m}$. Let $C_F$ be the  ternary linear code generated by Construction \ref{ourcon}. Then for any $c(\mu,\nu)\in C_F$, it holds that
\begin{align*}
{\rm wt}(c(\mu,\nu))=
\begin{cases}
0, \hspace{0.3cm}&{\rm if}~\mu={\bf0_m},\nu={\bf0_n},\\
3^n-3^{n-1}, &{\rm if}~\mu={\bf0_m},\nu\in\mathbb{F}_3^{n*},\\
3^n-3^{n-1}- \frac{2}{3}Re(W_F(\mu,-\nu)), &{\rm otherwise,}
\end{cases}
\end{align*}
where $Re(\xi)$ denotes the real part of the complex number $\xi$.
\end{proposition}
\begin{proof}
For any $(\mu,\nu)\in\mathbb{F}_{3}^m\times\mathbb{F}_3^n$, the weight of $c(\mu,\nu)\in C_F$ is that
\begin{align*}
{\rm wt}(c(\mu,\nu))=\#\{x\in \mathbb{F}_3^{n*}:\mu\cdot F(x)+\nu\cdot x\neq0\}.
\end{align*}
Since $F({\bf 0_n})={\bf 0_m}$, we have
\begin{align*}
{\rm wt}(c(\mu,\nu))=3^n-T,
\end{align*}
where $T=\#\{x\in\mathbb{F}_3^n:\mu\cdot F(x)+\nu\cdot x=0\}$. Note that
\begin{align*}
T=&\frac{1}{3}\sum_{x\in\mathbb{F}_{3}^n}\sum_{y\in\mathbb{F}_3}\zeta_3^{(\mu\cdot F(x)+\nu\cdot x)y}\\
=&\frac{1}{3}\sum_{x\in\mathbb{F}_{3}^n}\bigg(1+\zeta_3^{\mu\cdot F(x)+\nu\cdot x}+\zeta_3^{-(\mu\cdot F(x)+\nu\cdot x)}\bigg)\\
=&3^{n-1}+\frac{2}{3}\sum_{x\in\mathbb{F}_{3}^n}Re(\zeta_3^{\mu\cdot F(x)+\nu\cdot x}).
\end{align*}
Thus we have
\begin{align*}
{\rm wt}(c(\mu,\nu))=3^n-3^{n-1}-\frac{2}{3}Re\bigg(\sum_{x\in\mathbb{F}_{3}^n}\zeta_3^{\mu\cdot F(x)+\nu\cdot x}\bigg),
\end{align*}
which is equal to $0$ if $\mu={\bf 0_m}, \nu={\bf 0_n}$, equal to $3^n-3^{n-1}$ if $\mu={\bf 0_m}, \nu\neq{\bf 0_n}$, and equal to $3^n-3^{n-1}-\frac{2}{3}Re(W_F(\mu,-\nu))$ otherwise.
\end{proof}

\begin{remark}\label{rmkdistance}
Let $F$ be a function from $\mathbb{F}_3^n$ to $\mathbb{F}_3^m$ with $F({\bf0_n})={\bf0_m}$ and $\max_{\mu\in\mathbb{F}_3^{m*},\nu\in\mathbb{F}_3^n} Re\big(W_F(\mu,\nu)\big)$ $\geq0$. Then by Proposition \ref{weightdistri}, the minimal distance $d$ of the code $C_F$ in Construction \ref{ourcon} is that
\begin{align}\label{nonlinearity}
d=3^n-3^{n-1}-\frac{2}{3}\max_{\mu\in\mathbb{F}_3^{m*},\nu\in\mathbb{F}_3^n} Re\big(W_F(\mu,\nu)\big).
\end{align}
\end{remark}

\begin{remark}\label{rmk3-weight}
When $F$ is a function from $\mathbb{F}_3^n$ to $\mathbb{F}_3^m$ such that $F({\bf0_n})={\bf0_m}$ and $Re(W_F(\mu,\nu))\in\{0,N_1, N_2\}$ for any $(\mu,\nu)\in\mathbb{F}_3^{m*}\times\mathbb{F}_{3}^n$, then by Proposition \ref{weightdistri}, the linear code $C_F$ in Construction \ref{ourcon} is three-weight, whose weights are $3^{n}-3^{n-1}$, $3^{n}-3^{n-1}-\frac{2}{3}N_1$ and $3^{n}-3^{n-1}-\frac{2}{3}N_2$, respectively. In particular, if $F$ is a vectorial regular $s$-plateaued function, then $C_F$ is three-weight, since in this case $Re(W_F(\mu,\nu))\in\{0, 3^{(n+s)/2}, -\frac{3^{(n+s)/2}}{2}\}$ for any $(\mu,\nu)\in\mathbb{F}_3^{m*}\times\mathbb{F}_3^n$.  
\end{remark}

\begin{remark}\label{rmkABcondition}
Let $F$ be a function from $\mathbb{F}_{3}^n$ to $\mathbb{F}_{3}^m$ with $F({\bf 0_n})={\bf 0_m}$. If  $\max_{\mu\in\mathbb{F}_3^{m*},\nu\in\mathbb{F}_3^n}Re(W_F(\mu,\nu))$ $\geq 0$, $\min_{\mu\in\mathbb{F}_3^{m*},\nu\in\mathbb{F}_3^n}Re(W_F(\mu,\nu))\leq 0$, and
\begin{align}\label{eqrmkabcondition}
3\max_{\mu\in\mathbb{F}_3^{m*},\nu\in\mathbb{F}_3^n}Re(W_F(\mu,\nu))
-2\min_{\mu\in\mathbb{F}_3^{m*},\nu\in\mathbb{F}_3^n}Re(W_F(\mu,\nu))\geq 3^n,
\end{align}
then by Lemma \ref{lemABcondition} and Proposition \ref{weightdistri}, the linear code $C_F$ in Construction \ref{ourcon} does not satisfy the AB condition.
\end{remark}



\subsection{The minimality of $C_F$}

In this subsection, we determine the minimality of the linear codes $C_F$. 

\begin{theorem}\label{thminimality}
Let $F$ be a function from $\mathbb{F}_3^n$ to $\mathbb{F}_3^m$ such that $F({\bf 0_n})={\bf 0_m}$ and $\mu\cdot F$ is not affine for any $\mu\in\mathbb{F}_3^{m*}$. Then the linear code $C_F$ in Construction \ref{ourcon} is minimal if and only if the following two conditions are satisfied:
\begin{enumerate}
\item[(1)] For any $\mu\in\mathbb{F}_3^{m*}$, any pairwise distinct $\nu,\nu',\nu''\in\mathbb{F}_3^n$ with $\nu+ \nu'+\nu''={\bf 0_n}$, and any $\theta\in\{1,-2\}$, it holds that
\begin{align}\label{eqthmainresult1}
Re\big(W_F(\mu,\nu)+ W_F(\mu,\nu')+\theta W_F(\mu,\nu'')\big)\neq 3^n;
\end{align}

\item[(2)] For any $\mu,\mu'\in\mathbb{F}_{3}^{m*}$ with $\mu\neq\pm\mu'$, and any $\nu,\nu'\in\mathbb{F}_3^n$, it holds that
\begin{align}\label{eqthmainresult2}
Re\big(W_F(\mu,\nu)+W_F(\mu',\nu')+W_F(\mu+\mu',\nu+\nu')-2W_F(\mu-\mu',\nu-\nu')\big)\neq 3^n.
\end{align}
\end{enumerate}
\end{theorem}
\begin{proof}
 By assumption that $\mu\cdot F$ is not affine for any $\mu\in\mathbb{F}_3^{m*}$, it is easily seen that $c(\mu,\nu)\neq c(\mu',\nu')$ and $c(\mu,\nu)\pm c(\mu',\nu')=c(\mu\pm\mu',\nu\pm\nu')$ for any distinct $(\mu,\nu),(\mu',\nu')\in\mathbb{F}_3^m\times\mathbb{F}_3^n$. Then by Remark \ref{rmkq=3}, we need to show
\begin{align}\label{eqproofmain1}
{\rm wt}(c(\mu+\mu',\nu+\nu'))+{\rm wt}(c(\mu-\mu',\nu-\nu'))\neq 2{\rm wt}(c(\mu,\nu))-{\rm wt}(c(\mu',\nu'))
\end{align}
for any $(\mu,\nu),(\mu',\nu'),(\mu+\mu',\nu+\nu'),(\mu-\mu',\nu-\nu')\in\mathbb{F}_{3}^m\times\mathbb{F}_3^n
\backslash\{({\bf0_m},{\bf0_n})\}$. According to the relations of $\mu,\mu',\mu\pm\mu'$ with ${\bf 0_m}$, we have to distinguish following six cases.

{\bf Case 1}. $\mu=\mu'={\bf 0_m}$ and $\nu, \nu', \nu\pm\nu'\in\mathbb{F}_{3}^{n*}$. In this case, ${\rm wt}(c(\mu+\mu',\nu+\nu'))+{\rm wt}(c(\mu-\mu',\nu-\nu'))=2(3^n-3^{n-1})$ and $2{\rm wt}(c(\mu,\nu))-{\rm wt}(c(\mu',\nu'))=3^n-3^{n-1}$ by Proposition \ref{weightdistri}, which shows that \eqref{eqproofmain1} holds.

{\bf Case 2}.  $\mu={\bf0_m}, \mu'\neq {\bf 0_m}$ and $\nu\neq{\bf0_n}$. In this case,
${\rm wt}(c(\mu+\mu',\nu+\nu'))+{\rm wt}(c(\mu-\mu',\nu-\nu'))=2(3^n-3^{n-1})-\frac{2}{3}Re\big(W_F(\mu',-(\nu+\nu'))+W_F(-\mu',-\nu+\nu')\big)$,
$2{\rm wt}(c(\mu,\nu))-{\rm wt}(c(\mu',\nu'))=3^n-3^{n-1}+\frac{2}{3}Re\big(W_F(\mu',-\nu')\big)$ by Proposition \ref{weightdistri}. Note that $Re\big(W_F(-a,-b)\big)=Re\big(W_F(a,b)\big)$ for any $(a,b)\in\mathbb{F}_3^m\times\mathbb{F}_3^n$. Then \eqref{eqproofmain1} holds if and only if
\begin{align*}
Re\big(W_F(\mu',-(\nu+\nu'))+W_F(\mu',\nu-\nu')+W_F(\mu',-\nu')\big)\neq 3^n,
\end{align*}
where $-(\nu+\nu'),\nu-\nu'$ and $-\nu'$ are pairwise distinct (since $\nu\neq{\bf0_n}$) satisfying $-(\nu+\nu')+\nu-\nu'+(-\nu')={\bf0_n}$.

{\bf Case 3}. $\mu\neq {\bf 0_m}, \mu'={\bf 0_m}$ and $\nu'\neq {\bf 0_n}$. In this case,
${\rm wt}(c(\mu+\mu',\nu+\nu'))+{\rm wt}(c(\mu-\mu',\nu-\nu'))
=2(3^n-3^{n-1})-\frac{2}{3}Re\big(W_F(\mu,-(\nu+\nu'))+W_F(\mu,-\nu+\nu')\big)$,
$2{\rm wt}(c(\mu,\nu))-{\rm wt}(c(\mu',\nu'))=3^n-3^{n-1}-\frac{4}{3}Re\big(W_F(\mu,-\nu)\big)$ by Proposition \ref{weightdistri}. Then \eqref{eqproofmain1} holds if and only if
\begin{align*}
Re\big(W_F(\mu,-(\nu+\nu'))+W_F(\mu,-\nu+\nu')-2W_F(\mu,-\nu)\big)\neq 3^n,
\end{align*}
where $-(\nu+\nu'),-\nu+\nu'$ and $-\nu$ are pairwise distinct (since $\nu'\neq{\bf0_n}$) satisfying $-(\nu+\nu')+(-\nu+\nu')+(-\nu)={\bf0_n}$.

{\bf Case 4}. $\mu=\mu'\neq {\bf 0_m}$ and $\nu\neq\nu'$. In this case,
${\rm wt}(c(\mu+\mu',\nu+\nu'))+{\rm wt}(c(\mu-\mu',\nu-\nu'))
=2(3^n-3^{n-1})-\frac{2}{3}Re\big(W_F(\mu,\nu+\nu')\big)$,
$2{\rm wt}(c(\mu,\nu))-{\rm wt}(c(\mu',\nu'))=3^n-3^{n-1}-\frac{2}{3}Re\big(2W_F(\mu,-\nu)-W_F(\mu,-\nu')\big)$ by Proposition \ref{weightdistri}. Thus, \eqref{eqproofmain1} holds if and only if
\begin{align*}
Re\big(W_F(\mu,\nu+\nu')+W_F(\mu,-\nu')-2W_F(\mu,-\nu)\big)\neq 3^n,
\end{align*}
where $\nu+\nu',-\nu$ and $-\nu'$ are pairwise distinct (since $\nu\neq\nu'$) satisfying $\nu+\nu'+(-\nu)+(-\nu')={\bf0_n}$.

{\bf Case 5}. $\mu=-\mu'\neq {\bf 0_m}$ and $\nu\neq-\nu'$. In this case,
${\rm wt}(c(\mu+\mu',\nu+\nu'))+{\rm wt}(c(\mu-\mu',\nu-\nu'))
=2(3^n-3^{n-1})-\frac{2}{3}Re\big(W_F(\mu,\nu-\nu')\big)$,
$2{\rm wt}(c(\mu,\nu))-{\rm wt}(c(\mu',\nu'))=3^n-3^{n-1}-\frac{2}{3}Re\big(2W_F(\mu,-\nu)-W_F(\mu,\nu')\big)$ by Proposition \ref{weightdistri}. Thus, \eqref{eqproofmain1} holds if and only if
\begin{align*}
Re\big(W_F(\mu,\nu-\nu')+W_F(\mu,\nu')-2W_F(\mu,-\nu)\big)\neq 3^n,
\end{align*}
where $\nu-\nu',-\nu$ and $\nu'$ are pairwise distinct (since $\nu\neq-\nu'$) satisfying $\nu-\nu'+(-\nu)+(\nu')={\bf0_n}$.

{\bf Case 6}. $\mu\neq {\bf 0_m}, \mu'\neq {\bf 0_m}$ and $\mu\pm\mu'\neq{\bf 0_m}$. In this case,
${\rm wt}(c(\mu+\mu',\nu+\nu'))+{\rm wt}(c(\mu-\mu',\nu-\nu'))
=2(3^n-3^{n-1})-\frac{2}{3}Re\big(W_F(\mu+\mu',-(\nu+\nu'))+W_F(\mu-\mu',-\nu+\nu')\big)$,
$2{\rm wt}(c(\mu,\nu))-{\rm wt}(c(\mu',\nu'))=3^n-3^{n-1}-\frac{2}{3}Re\big(2W_F(\mu,-\nu)-W_F(\mu',-\nu')\big)$ by Proposition \ref{weightdistri}. Then \eqref{eqproofmain1} holds if and only if
\begin{align*}
Re\big(W_F(\mu+\mu',-(\nu+\nu'))+W_F(\mu-\mu',-\nu+\nu')+W_F(\mu',-\nu')-2W_F(\mu,-\nu)\big)\neq 3^n.
\end{align*}
Let $a=\mu+\mu'$, $a'=\mu'$, $b=-(\nu+\nu')$ and $b'=-\nu'$. Then we have $\mu-\mu'=a+a'$, $\mu=a-a'$, $-\nu+\nu'=b+b'$ and $-\nu=b-b'$, and thus the above relation becomes that
\begin{align*}
Re\big(W_F(a,b)+W_F(a+a',b+b')+W_F(a',b')-2W_F(a-a',b-b')\big)\neq 3^n,
\end{align*}
where $a,a',a\pm a'\in\mathbb{F}_3^{m*}$, since $\mu,\mu',\mu\pm\mu'\in\mathbb{F}_3^{m*}$.

Therefore, the result follows from the above six cases.
\end{proof}

\begin{remark}
When $m=1$, there do not exist $\mu,\mu'\in\mathbb{F}_3^*$ such that $\mu\pm\mu'\in\mathbb{F}_3^*$, and hence Condition (2) of Theorem \ref{thminimality} does not exist. In this case, Theorem \ref{thminimality} is reduced to
 \cite[Theorem 14]{Heng-2018-FFA}.
\end{remark}

According to Theorem \ref{thminimality} and the fact that $Re(W_F(\mu,\nu))\leq \big|Re(W_F(\mu,\nu))|\leq |W_F(\mu,\nu)\big|$ and $\big|Re(W_F(\mu,\nu))\pm Re(W_F(\mu',\nu'))\big|\leq\big|W_F(\mu,\nu)|+|W_F(\mu',\nu')\big|$ for any $(\mu,\nu),(\mu',\nu')\in\mathbb{F}_3^{m}\times\mathbb{F}_3^n$, it is easy to deduce the following result.

\begin{corollary}\label{rmkminimal}
Let $F$ be a function from $\mathbb{F}_3^n$ to $\mathbb{F}_3^m$ such that $F({\bf 0_n})={\bf0_m}$ and $\big|W_F(\mu,\nu)\big|<\frac{3^n}{5}$ for any $(\mu,\nu)\in\mathbb{F}_3^{m*}\times\mathbb{F}_3^n$. Then the linear code $C_F$ in construction \ref{ourcon} is minimal. 
\end{corollary}

\section{Three-weight  minimal ternary linear codes from vectorial regular $s$-plateaued functions}\label{sec:threeweight}


In this section, we present some three-weight minimal ternary linear codes by using certain vectorial regular $s$-plateaued functions. To this end, we first give the following characterization of three-weight minimal ternary linear codes.

\begin{theorem}\label{ththreewcha}
Let $F$ be a vectorial function from $\mathbb{F}_3^n$ to $\mathbb{F}_3^m$ satisfying the following two conditions:
\begin{enumerate}
\item[(1)]  $F({\bf0_n})={\bf0_m}$;

\item[(2)] $Re(W_F(\mu,\nu))\in\{0,N_1,N_2\}$ for any $(\mu,\nu)\in\mathbb{F}_3^{m*}\times\mathbb{F}_3^n$ such that $\max\{|N_1|,|N_2|\}<\frac{3^n}{5}$, where $N_1$ and $N_2$ are two rational numbers.
\end{enumerate}
 Then the linear code $C_F$ generated by Construction \ref{ourcon} is a three-weight minimal ternary linear code with parameters $[3^n-1,n+m,d]$, where $d=3^n-3^{n-1}-\frac{2}{3}\max\{0, N_1, N_2\}$.
\end{theorem}

\begin{proof}
Note that Condition (2) of this theorem implies that $\mu\cdot F$ is not affine for any $\mu\in\mathbb{F}_2^{m*}$. Then by Proposition \ref{propdim} and Remark \ref{rmk3-weight}, we obtain that $C_F$ is a three-weight linear code with parameters $[3^n-1,n+m,d]$ if $Re(W_F(\mu,\nu))\in\{0,N_1,N_2\}$ for any $(\mu,\nu)\in\mathbb{F}_3^{m*}\times\mathbb{F}_3^n$. By 
Corollary \ref{rmkminimal}, we derive that $C_F$ is minimal.
\end{proof}

Applying Theorem \ref{ththreewcha} to vectorial  regular $s$-plateaued functions, we deduce the following result.

\begin{theorem}\label{th3weirgl}
Let $n>4$ and $0\leq s\leq n-4$ be two integers with the same parity. Let $F:\mathbb{F}_3^n\to\mathbb{F}_3^m$ be a vectorial regular $s$-plateaued function with $F({\bf0_n})={\bf0_m}$. Then the linear code $C_F$ generated by Construction \ref{ourcon} is a three-weight minimal linear code with parameters $[3^n-1,n+m,3^n-3^{n-1}-3^{\frac{n+s}{2}}+3^{\frac{n+s}{2}-1}]$. Moreover, the weight distribution of $C_F$ is given in Table \ref{b2}.
\begin{table*}[!htbp]\setlength{\abovecaptionskip}{0cm}
\caption{Weight distribution of $C_F$ in Theorem \ref{th3weirgl}}  \centering \label{b2}
\medskip
\begin{tabular}{|c|c|}
 \hline
 Weight & Frequency\\
 \hline
 0 & 1\\
 \hline
 $3^n-3^{n-1}$ & $3^{n}-1+(3^m-1)(3^n-3^{n-s})$\\
 \hline
 $3^n-3^{n-1}-3^{\frac{n+s}{2}}+3^{\frac{n+s}{2}-1}$ & $(3^m-1)(3^{n-s-1}+3^{\frac{n-s}{2}}-3^{\frac{n-s}{2}-1})$\\
 \hline
 $3^n-3^{n-1}+3^{\frac{n+s}{2}-1}$ & $2(3^m-1)(3^{n-s-1}-3^{\frac{n-s}{2}-1})$\\
 \hline
\end{tabular}
\end{table*}
\end{theorem}
\begin{proof}
Since $F$ is a vectorial regular $s$-plateaued function, we have
$$W_F(\mu,\nu)\in\big\{0, 3^{\frac{n+s}{2}}, 3^{\frac{n+s}{2}}\zeta_3,3^{\frac{n+s}{2}}\zeta_3^2\big\}$$
 for any $(\mu,\nu)\in\mathbb{F}_3^{m*}\times\mathbb{F}_3^n$, which implies that $Re(W_F(\mu,\nu))\in\{0,3^{\frac{n+s}{2}}, -\frac{3^{\frac{n+s}{2}}}{2}\}$ for any $(\mu,\nu)\in\mathbb{F}_3^{m*}\times\mathbb{F}_3^n$, since $Re(\zeta_3)=Re(\zeta_3^2)=-\frac{1}{2}$. When $s\leq n-4$, we have $5\times 3^{\frac{n+s}{2}}\leq 5\times 3^{n-2}<3^n$. Then by Theorem \ref{ththreewcha}, $C_F$ is a three-weight minimal linear code with parameters $[3^n-1,n+m,d]$, where $d=3^n-3^{n-1}-\frac{2}{3}\times 3^{\frac{n+s}{2}}=3^n-3^{n-1}- 3^{\frac{n+s}{2}}+ 3^{\frac{n+s}{2}-1}$. Below, we determine the weight distribution of $C_F$.

 Since  $Re(W_F(\mu,\nu))\in\{0,3^{\frac{n+s}{2}}, -\frac{3^{\frac{n+s}{2}}}{2}\}$ for any $(\mu,\nu)\in\mathbb{F}_3^{m*}\times\mathbb{F}_3^n$,   we obtain from Proposition \ref{weightdistri} that
 \begin{align*}
 {\rm wt}(c(\mu,\nu))\in T:=\{0, d_0,d, d_1 \}
 \end{align*}
for any $(\mu,\nu)\in\mathbb{F}_3^{m}\times\mathbb{F}_3^n$, where $d_0=3^n-3^{n-1}$, $d=3^n-3^{n-1}- 3^{\frac{n+s}{2}}+ 3^{\frac{n+s}{2}-1}$, and $d_1=3^n-3^{n-1}+ 3^{\frac{n+s}{2}-1}$. Therefore, we need to  determine the frequencies of all elements in the set $T$, which are denoted by $\mathbf{f}(0)$, $\mathbf{f}(d_1)$,
 $\mathbf{f}(d)$ and $\mathbf{f}(d_1)$, respectively.

 Let $\mu$ be a fixed element of $\mathbb{F}_3^{m*}$. Let $S_F:=\big\{\nu\in\mathbb{F}_3^n: W_F(\mu,\nu)\neq 0\big\}$. Then from Parseval's identity $\sum_{\nu\in\mathbb{F}_3^n}|W_F(\mu,\nu)|^2=3^{2n}$ and $|W_F(\mu,\nu)|=3^{\frac{n+s}{2}}$ for any $\nu\in S_F$, we obtain that $\#S_F=3^{n-s}$ and $\#\overline{S_F}=3^n-3^{n-s}$, where $\overline{S_F}=\mathbb{F}_3^n\backslash S_F$. Let
$S_F^{(i)}:=\big\{\nu\in\mathbb{F}_3^n: W_F(\mu,\nu)= 3^{\frac{n+s}{2}}\zeta_3^i\big\}$ for $i\in\mathbb{F}_3$. Then we have
\begin{align}\label{eqth3w1}
\#S_F^{(0)}+\#S_F^{(1)}+\#S_F^{(2)}=3^{n-s}.
\end{align}
Besides, according to the inverse Walsh transform of $\mu\cdot F$, we have
\begin{align}\label{eqinwalsh0}
\sum_{\nu\in\mathbb{F}_3^n}W_F(\mu,\nu)=3^n\zeta_3^{\mu\cdot F({\mathbf{0_n}})}=3^n,
\end{align}
which implies that
\begin{align}\label{eqth3w2}
3^{\frac{n+s}{2}}\big(\#S_F^{(0)}+\#S_F^{(1)}\zeta_3+\#S_F^{(2)}\zeta_3^2\big)=3^n.
\end{align}
Note that
\begin{align*}
\#S_F^{(1)}\zeta_3+\#S_F^{(2)}\zeta_3^2
=&\#S_F^{(1)}(-\frac{1}{2}+\frac{\sqrt{-3}}{2})+\#S_F^{(2)}(-\frac{1}{2}-\frac{\sqrt{-3}}{2})\\
=&-\frac{1}{2}\big(\#S_F^{(1)}+\#S_F^{(2)}\big)+\frac{\sqrt{-3}}{2}\big(\#S_F^{(1)}-\#S_F^{(2)}\big).
\end{align*}
Then from \eqref{eqth3w2}, we obtain that $\#S_F^{(1)}=\#S_F^{(2)}$ and
\begin{align}\label{eqth3w3}
\#S_F^{(0)}-\#S_F^{(1)}=3^{\frac{n-s}{2}}.
\end{align}
Combining \eqref{eqth3w1} with \eqref{eqth3w3}, we derive that
\begin{align*}
\#S_F^{(0)}=3^{n-s-1}+3^{\frac{n-s}{2}}-3^{\frac{n-s}{2}-1} \hspace{0.3cm} {\rm and} \hspace{0.3cm} \#S_F^{(1)}=\#S_F^{(2)}=3^{n-s-1}-3^{\frac{n-s}{2}-1}.
\end{align*}

Therefore, we obtain from Proposition \ref{weightdistri} that
\begin{align*}
&\mathbf{f}(0)=1, \\
&\mathbf{f}(d_0)=3^n-1+(3^m-1)\#\overline{S_F}=3^n-1+(3^m-1)(3^n-3^{n-s}),\\
&\mathbf{f}(d)=(3^m-1)\#S_F^{(0)}=(3^m-1)(3^{n-s-1}+3^{\frac{n-s}{2}}-3^{\frac{n-s}{2}-1}), \\
&\mathbf{f}(d_1)=(3^m-1)(\#S_F^{(1)}+\#S_F^{(2)})=2(3^m-1)(3^{n-s-1}-3^{\frac{n-s}{2}-1}).
\end{align*}
 This completes the proof.
\end{proof}

If $s=0$ or $s=1$, then we get the following two corollaries by Theorem \ref{th3weirgl}.

\begin{corollary}\label{cors=0}
Let $n>4$ be an even positive integer and $1\leq m\leq \frac{n}{2}$. Let $F:\mathbb{F}_3^n\to\mathbb{F}_3^m$ be a vectorial regular bent function with $F({\bf0_n})={\bf0_m}$. Then the code $C_F$ defined by Construction \ref{ourcon} is a three-weight minimal $[3^n-1, n+m,3^n-3^{n-1}-3^{\frac{n}{2}}+3^{\frac{n}{2}-1}]$ linear code. Moreover, the wight distribution of $C_F$ is given in Table \ref{b3}
\begin{table*}[!htbp]\setlength{\abovecaptionskip}{0cm}
\caption{Weight distribution of $C_F$ in Corollary \ref{cors=0}}  \centering \label{b3}
\medskip
\begin{tabular}{|c|c|}
 \hline
 Weight & Frequency\\
 \hline
 0 & 1\\
 \hline
 $3^n-3^{n-1}$ & $3^{n}-1$\\
 \hline
 $3^n-3^{n-1}-3^{\frac{n}{2}}+3^{\frac{n}{2}-1}$ & $(3^m-1)(3^{n-1}+3^{\frac{n}{2}}-3^{\frac{n}{2}-1})$\\
 \hline
 $3^n-3^{n-1}+3^{\frac{n}{2}-1}$ & $2(3^m-1)(3^{n-1}-3^{\frac{n}{2}-1})$\\
 \hline
\end{tabular}
\end{table*}
\end{corollary}


\begin{corollary}\label{cors=1}
Let $n>4$ be an odd positive integer. Let $F:\mathbb{F}_3^n\to\mathbb{F}_3^m$ be a vectorial regular $1$-plateaued function with $F({\bf0_n})={\bf0_m}$. Then the code $C_F$ defined by Construction \ref{ourcon} is a three-weight minimal  linear code with parameters $[3^n-1, n+m,3^n-3^{n-1}-3^{\frac{n+1}{2}}+3^{\frac{n+1}{2}-1}]$. Moreover, the wight distribution of $C_F$ is given in Table \ref{b4}.
\begin{table*}[!htbp]\setlength{\abovecaptionskip}{0cm}
\caption{Weight distribution of $C_F$ in Corollary \ref{cors=1}}  \centering \label{b4}
\medskip
\begin{tabular}{|c|c|}
 \hline
 Weight & Frequency\\
 \hline
 0 & 1\\
 \hline
 $3^n-3^{n-1}$ & $3^{n}-1+(3^m-1)(3^n-3^{n-1})$\\
 \hline
 $3^n-3^{n-1}-3^{\frac{n+1}{2}}+3^{\frac{n+1}{2}-1}$ & $(3^m-1)(3^{n-2}+3^{\frac{n-1}{2}}-3^{\frac{n-1}{2}-1})$\\
 \hline
 $3^n-3^{n-1}+3^{\frac{n+1}{2}-1}$ & $2(3^m-1)(3^{n-2}-3^{\frac{n-1}{2}-1})$\\
 \hline
\end{tabular}
\end{table*}
\end{corollary}


It is noted in Theorem \ref{th3weirgl} that $w_{\min}=3^n-3^{n-1}-3^{\frac{n+s}{2}}+3^{\frac{n+s}{2}-1}$ and $w_{\max}=3^n-3^{n-1}+3^{\frac{n+s}{2}-1}$, which satisfy $\frac{w_{\min}}{w_{\max}}>\frac{2}{3}$ for any $n-s\geq 4$. Therefore, the minimal codes obtained from this section satisfy the AB condition.

In literatures, tremendous efforts  (see, e.g., \cite{Ding-2015-IT, Ding-2013-IT, Ding-2015-IT2, Sihem-2019, Sihem-2020-IT2, Pelen-2022-AMC, Rodriguez-2023-IT, Tangcm-2016-IT, Xia-2017-FFA, Xu-2019-IT, Yang-2017-DCC, Zhou-2018-DCC, Zhouzc-2016-DCC}) have been paid to find linear codes with few weights, because of their important applications in data storage systems, communication systems and consumer electronics. At the end of this section, we collect the parameters of some known three-weight ternary linear codes in Table \ref{b3-weight}, from which one can easily see that the dimensions and minimal distances of the ternary linear codes given by Theorems \ref{th3weirgl}, Corollaries \ref{cors=0} and \ref{cors=1} are better than many of the others.

\begin{table*}[t]\setlength{\abovecaptionskip}{0cm}
 \small
\caption{The parameters of some known three-weight ternary linear codes}  \centering \label{b3-weight}
\medskip
\begin{tabular}{|c|c|c|c|}
 \hline
Length& Dimension & Minimal distance & References\\

 \hline
$\frac{3^{3n-1}-1}{2}$ &  $3n$  &$3^{3n-2}-3^{2n-2}$          &{\cite[Theorem 16]{Ding-2015-IT}}\\
 \hline
$3^n-1$ &  $2n$  &$2\times3^{n-1}-3^{\frac{n-1}{2}}$          &{\cite[Theorems 4.3, 4.8]{Ding-2013-IT}}\\
\hline
$3^n-1$ &  $2n$  &$2\times(3^{n-1}-3^{\frac{n-1}{2}})$          &{\cite[Theorem 5.2]{Ding-2013-IT}}\\
\hline
$3^{n-1}-1$    &  $n$   &$2\times(3^{n-2}-3^{\frac{n-3}{2}})$ &{\cite[Theorem 1]{Ding-2015-IT2}}\\
\hline
$\frac{3^{n-1}-1}{2}$    &  $n$   &$3^{n-2}-3^{\frac{n-3}{2}}$ &{\cite[Corallary 3]{Ding-2015-IT2}}\\
\hline
 $3^n-1$      & $n+1$ & $2\!\times\!(3^{n-1}\!-\! 3^{\frac{n+s-2}{2}})$ or
  $3^n\!-\!3^{n-1}\!-\! 3^{\frac{n+s-2}{2}}$                       & {\cite[Theorem 3]{Sihem-2019}} \\
 \hline
 $3^n-1$      & $n+1$ &  $3^n-3^{n-1}- 3^{\frac{n+s-1}{2}}$ & {\cite[Theorem 4]{Sihem-2019}} \\
 \hline
 $3^{n-1}-1$  & $n$   & $2\times(3^{n-2}-3^{\frac{n+s-3}{2}})$ & {\cite[Theorem 1]{Sihem-2020-IT2}} \\
  \hline
 $3^{r-1}\!-\!3^{\frac{n}{2}-1}\!+\!3^{\frac{n}{2}}\!-\!1$ & $r(r\!\geq\!\frac{n}{2}\!+\!1)$  &   $2\times3^{r-2}$
                                                             &  {\cite[Theorem 3.1]{Pelen-2022-AMC}}\\
 \hline
  $3^{r-1}+3^{\frac{n-1}{2}}$ & $r(r\geq\frac{n+1}{2})$  &   $2\times3^{r-2}$
                                                             &  {\cite[Theorem 3.2]{Pelen-2022-AMC}}\\
  \hline
  $3^{r-1}+3^{\frac{n}{2}-1}$ & $r(r\!\geq\!\frac{n}{2}\!+\!1)$  &   $2\times3^{r-2}$
                                                             &  {\cite[Theorem 4.1]{Pelen-2022-AMC}}\\
 \hline
  $3^n-1$     & $n+1$ &  $3^n-3^{n-1}-3^{\frac{n+k-1}{2}} $ &   
                                                     {\cite[Theorems 4, 5]{Rodriguez-2023-IT}}  \\
  \hline
$3^{n-1}-1$            &  $n$   &$2\times(3^{n-2}-3^{\frac{n-3}{2}})$  &{\cite[Theorem 16]{Tangcm-2016-IT}}\\
\hline
$\frac{3^{n-1}-1}{2}$            &  $n$   &$3^{n-2}-3^{\frac{n-3}{2}}$  &{\cite[Corollary 18]{Tangcm-2016-IT}}\\
\hline
$3^{n-1}-1$            &  $n$   &$2\times(3^{n-2}-3^{\frac{n+e}{2}-2})$  &{\cite[Theorem 1]{Xia-2017-FFA}}\\
\hline
$3^{n-1}-1$            &  $n$   &$2\times3^{n-2}-3^{\frac{n+e}{2}-2}$  &{\cite[Theorem 2]{Xia-2017-FFA}}\\
\hline
 $3^n-1$      & $n$   & $2^2\times 3^{n-2}$                 & {\cite[Theorem 1]{Xu-2019-IT}} \\
\hline
$3^{n-1}-1$            & $n$    &$2\!\times\!(3^{n-2}\!-\!3^{\frac{n-2}{2}})$ or $2\!\times\!(3^{n-2}\!-\!3^{\frac{n-3}{2}})$                                                             & {\cite[Theprem 2]{Yang-2017-DCC}}\\
\hline
$\frac{3^{sk-1}}{2},s=2l+1$ & $sk$   &$3^{sk-2}-3^{(l+1)k-2}$    &{\cite[Theorem 3]{Zhou-2018-DCC}}\\
 \hline
$3^{n-1}-1$         &  $n$   & $2\times(3^{n-2}-3^{\frac{n-3}{2}})$ &{\cite[Theorem 1]{Zhouzc-2016-DCC}}\\
\hline
$\frac{3^{n-1}-1}{2}$         &  $n$   & $3^{n-2}-3^{\frac{n-3}{2}}$ &{\cite[Corollary 6]{Zhouzc-2016-DCC}}\\
 \hline
 $3^{n}-1$ &   $n+m$ &  $3^n-3^{n-1}-3^{\frac{n+s}{2}}+3^{\frac{n+s}{2}-1}$  & Theorem \ref{th3weirgl}\\
 \hline
 $3^{n}-1$ &   $n+m$ &  $3^n-3^{n-1}-3^{\frac{n}{2}}+3^{\frac{n}{2}-1}$  & Corollary \ref{cors=0} \\
 \hline
 $3^{n}-1$ &   $n+m$ &  $3^n-3^{n-1}-3^{\frac{n+1}{2}}+3^{\frac{n+1}{2}-1}$  & Corollary \ref{cors=1} \\
 \hline

\end{tabular}
\end{table*}

\section{A generic construction of minimal ternary linear codes going against the AB condition}\label{sec:mvab}

 Similarly as that of Section V of \cite{LYJetal-2023-IT}, in this section, we shall exhibit a generic construction of minimal ternary linear codes going against the AB condition. We first define the following construction of ternary linear codes.

 \begin{construction}\label{con2}
 Let $f$ be a function from $\mathbb{F}_3^n$ to $\mathbb{F}_3$ such that $f({\bf0_n})=0$, $\max_{\nu\in\mathbb{F}_3^n} Re\big(W_f(\nu)\big)$ $\geq0$, $\min_{\nu\in\mathbb{F}_p^n}Re(W_f(\nu))\leq 0$, and the following two conditions are fulfilled:
\begin{enumerate}
\item[{\rm({\bf a})}] $Re\big(W_f(\nu)+ W_f(\nu')+\theta W_f(\nu'')\big)\neq 3^n$ for any pairwise distinct $\nu$,$\nu'$,$\nu''\in\mathbb{F}_3^n$ with $\nu+ \nu'+\nu''={\bf 0_n}$, and for any $\theta\in\{1,-2\}$;

\item[{\rm({\bf b})}] $3\max_{\nu\in\mathbb{F}_3^n}Re(W_f(\nu))
-2\min_{\nu\in\mathbb{F}_3^n}Re(W_f(\nu))\geq 3^n$.
 \end{enumerate}
 Let $G$ be a function from $\mathbb{F}_3^n$ to $\mathbb{F}_3^{m-1}$ such that $G(\mathbf{0_n})=\mathbf{0_{m-1}}$,  $\mu_1f+\tilde{\mu}\cdot G$ is not affine for all nonzero $\mu=(\mu_1,\tilde{\mu})\in\mathbb{F}_3\times\mathbb{F}_3^{m-1}$, and  the linear code $C_G$ defined by
 \begin{align}\label{conG}
 C_G=\bigg\{c(\tilde{\mu},\nu)=\big(\tilde{\mu}\cdot G(x)+\nu\cdot x\big)_{x\in\mathbb{F}_3^{n*}} : (\tilde{\mu},\nu)\in\mathbb{F}_3^{m-1}\times \mathbb{F}_3^n\bigg\}
 \end{align}
 is minimal.
Let $F$ be the function from $\mathbb{F}_3^n$ to $\mathbb{F}_3^m$ defined as
\begin{align}\label{mainfunction}
F(x)=(f(x), G(x)),
 \end{align}
 and let $C_F$ be the linear code generated by Construction \ref{ourcon}.
 \end{construction}

  Then we find some conditions  such that the ternary linear code $C_F$ generated by Construction \ref{con2} is minimal and violates the AB condition. Similarly as that of \cite[Proposition 4]{LYJetal-2023-IT}, we need first  to determine the Walsh transform of $F$.

 \begin{proposition}\label{propwalshtransform}
Let $F$ be the function from $\mathbb{F}_3^n$ to $\mathbb{F}_3^m$ defined as in Construction \ref{con2}. Then the Walsh transform of $F$ at $\mu=(\mu_1,\tilde{\mu})\in\mathbb{F}_3\times\mathbb{F}_3^{m-1}\backslash\{(0,{\bf0_{m-1}})\}$, $\nu\in\mathbb{F}_3^n$ is given by
\begin{align*}
W_F(\mu,\nu)=\begin{cases}
W_G(\tilde{\mu},\nu),\hspace{0.3cm}&{\rm if} ~\mu_1=0,~\tilde{\mu}\neq{\mathbf{0_{m-1}}},\\
W_f(\nu),\hspace{0.3cm}&{\rm if} ~\mu_1=1,~\tilde{\mu}={\mathbf{0_{m-1}}},\\
W_{A_{\tilde{\mu}}}(\nu), \hspace{0.3cm}&{\rm if} ~\mu_1=1,~\tilde{\mu}\neq{\mathbf{0_{m-1}}},\\
\overline{W_{f}(-\nu)},\hspace{0.3cm}&{\rm if} ~\mu_1=-1,~\tilde{\mu}={\mathbf{0_{m-1}}},\\
\overline{W_{A_{2\tilde{\mu}}}(-\nu)}, \hspace{0.3cm}&{\rm if} ~\mu_1=-1,~\tilde{\mu}\neq{\mathbf{0_{m-1}}},
\end{cases}
\end{align*}
where $A_{{\tilde{\mu}}}(x)=f(x)+\tilde{\mu}\cdot G(x)$.
\end{proposition}
 \begin{proof}
Recall that $F=(f, G)$ in Construction \ref{con2}. For any $\mu=(\mu_1,\tilde{\mu})\in\mathbb{F}_3\times\mathbb{F}_3^{m-1}\backslash\{(0,{\bf0_{m-1}})\}$ and $\nu\in\mathbb{F}_3^n$, one has
\begin{align*}
W_F(\mu,\nu)=\sum_{x\in\mathbb{F}_3^n}\zeta_3^{\mu\cdot F(x)-\nu\cdot x}=\sum_{x\in\mathbb{F}_3^n}\zeta_3^{\mu_1f(x)+\tilde{\mu}\cdot G(x)-\nu\cdot x}.
\end{align*}
Then the result follows from the fact that $\zeta_3^{-j}=\overline{\zeta_3^j}$ for any $j\in\mathbb{F}_3$.
\end{proof}

From Proposition \ref{propwalshtransform}, we know that
\begin{align*}
\max_{\mu\in\mathbb{F}_3^{m*},\nu\in\mathbb{F}_3^n}Re(W_F(\mu,\nu))\geq \max_{\nu\in\mathbb{F}_3^n}Re(W_f(\nu))\hspace{0.2cm}{\rm and}\hspace{0.2cm}
\min_{\mu\in\mathbb{F}_3^{m*},\nu\in\mathbb{F}_3^n}Re(W_F(\mu,\nu))\leq \min_{\nu\in\mathbb{F}_3^n}Re(W_f(\nu)),
\end{align*}
and hence according to Condition ({\bf b}) of Construction \ref{con2}, one has
\begin{align*}
3\max_{\mu\in\mathbb{F}_3^{m*},\nu\in\mathbb{F}_3^n}Re(W_F(\mu,\nu))
-2\min_{\mu\in\mathbb{F}_3^{m*},\nu\in\mathbb{F}_3^n}Re(W_F(\mu,\nu))\geq 3^n.
\end{align*}
Then by Remark \ref{rmkABcondition}, we deduce the following proposition.

\begin{proposition}\label{propviolateAB}
The ternary linear code $C_F$ generated by Construction \ref{con2} does not satisfy the AB condition.
\end{proposition}

In the following theorem, we give a necessary and sufficient condition such that the ternary linear code $C_F$ in Construction \ref{con2} is minimal.

\begin{theorem}\label{thNS}
The ternary linear code $C_F$ defined by Construction \ref{con2} is minimal if and only if the following three conditions are satisfied:
\begin{enumerate}
\item[(1)] For any $\tilde{\mu}\in\mathbb{F}_3^{m-1}\backslash\{\mathbf{0_{m-1}}\}$, any pairwise distinct $\nu,\nu',\nu''\in\mathbb{F}_3^n$ with $\nu+ \nu'+\nu''={\bf0_n}$, and any $\theta\in\{1,-2\}$, it holds that
\begin{align}\label{eqgAB1}
Re\big(W_{A_{\tilde{\mu}}}(\nu)+W_{A_{\tilde{\mu}}}(\nu')+\theta W_{A_{\tilde{\mu}}}(\nu'')\big)\neq 3^{n};
\end{align}

\item[(2)] For any $\tilde{\mu}\in\mathbb{F}_3^{m-1}\backslash\{\mathbf{0_{m-1}}\}$, any $\nu,\nu'\in\mathbb{F}_3^n$, and any $\lambda_1,\lambda_2,\lambda_3,\lambda_4\in \{1, -2\}$ with one of them being $-2$,  it holds that
\begin{align}\label{eqgAB2}
Re\big(\lambda_1W_f(\nu)+\lambda_2W_G(\tilde{\mu},\nu')+\lambda_3W_{A_{\tilde{\mu}}}(\nu+\nu')
+\lambda_4W_{A_{2\tilde{\mu}}}(\nu-\nu')\big)\neq3^n;
\end{align}

\item[(3)] For any  $\tilde{\mu},\tilde{\mu}',\tilde{\mu}\pm\tilde{\mu}'\in\mathbb{F}_3^{m-1}\backslash\{\mathbf{0_{m-1}}\}$,  any $\nu,\nu'\in\mathbb{F}_3^n$,  and any $\lambda_1,\lambda_2,\lambda_3,\lambda_4\in \{1, -2\}$ with one of them being $-2$, it holds that
    \begin{align}\label{eqgAB3}
Re\big(\lambda_1W_{A_{\tilde{\mu}}}(\nu)+\lambda_2W_G(\tilde{\mu}',\nu')
+\lambda_3W_{A_{\tilde{\mu}+\tilde{\mu}'}}(\nu+\nu')
+\lambda_4W_{A_{\tilde{\mu}-\tilde{\mu}'}}(\nu-\nu')\big)\neq 3^n.
\end{align}
\end{enumerate}
\end{theorem}
\begin{proof}
To proof this theorem, we need to check the two conditions of Theorem \ref{thminimality}. We first verify the first condition of Theorem \ref{thminimality}.  For any $\mu=(\mu_1,\tilde{\mu})\in\mathbb{F}_3\times\mathbb{F}_3^{m-1}\backslash\{(0,{\bf0_{m-1}})\}$, any pairwise distinct $\nu,\nu',\nu''\in\mathbb{F}_3^n$ with $\nu+\nu'+\nu''={\bf0_n}$, and any $\theta\in\{1,-2\}$, let $\Delta=Re\big(W_F(\mu,\nu)+W_F(\mu,\nu')+\theta W_F(\mu,\nu'')\big)$. Then by Proposition \ref{propwalshtransform}, one has
\begin{align*}
\Delta
=\begin{cases}
Re\big(W_G(\tilde{\mu},\nu)+W_G(\tilde{\mu},\nu')
+\theta W_G(\tilde{\mu},\nu'')\big),&{\rm if}~\mu=(0,\tilde{\mu}),\\
Re\big(W_f(\nu)+W_f(\nu')+\theta W_f(\nu'')\big),~~&{\rm if}~\mu=(1,{\bf0_{m-1}}),\\
Re\big(W_{A_{\tilde{\mu}}}(\nu)+W_{A_{\tilde{\mu}}}(\nu')+\theta W_{A_{\tilde{\mu}}}(\nu'')\big),
&{\rm if}~\mu=(1,\tilde{\mu}),\tilde{\mu}\neq{\bf0_{m-1}},\\
Re\big(W_f(-\nu)+W_f(-\nu')+\theta W_f(-\nu'')\big),~~&{\rm if}~\mu=(-1,{\bf0_{m-1}}),\\
Re\big(W_{A_{2\tilde{\mu}}}(-\nu)+W_{A_{2\tilde{\mu}}}(-\nu')+\theta W_{A_{2\tilde{\mu}}}(-\nu'')\big),
&{\rm if}~\mu=(-1,\tilde{\mu}),\tilde{\mu}\neq{\bf0_{m-1}}.
\end{cases}
\end{align*}
By Condition ({\bf a}) of Construction \ref{con2}, we have $\Delta\neq 3^n$ if $\mu=(1,{\bf0_{m-1}})$ and $\mu=(-1,{\bf0_{m-1}})$. By Theorem \ref{thminimality} and the minimality of $C_G$ in Construction \ref{con2}, we have $\Delta\neq 3^n$ if $\mu=(0,\tilde{\mu})$. Hence, for any $\mu=(\mu_1,\tilde{\mu})\in\mathbb{F}_3\times\mathbb{F}_3^{m-1}\backslash\{(0,{\bf0_{m-1}})\}$, any pairwise distinct $\nu,\nu',\nu''\in\mathbb{F}_3^n$ with $\nu+\nu'+\nu''={\bf0_n}$, and any $\theta\in\{1,-2\}$, $\Delta\neq 3^n$ if and only if for any $\tilde{\mu}\in\mathbb{F}_3^{m-1}\backslash\{{\bf0_{m-1}}\}$,  any pairwise distinct $\nu,\nu',\nu''\in\mathbb{F}_3^n$ with $\nu+\nu'+\nu''={\bf0_n}$, and any $\theta\in\{1,-2\}$, $Re\big(W_{A_{\tilde{\mu}}}(\nu)+W_{A_{\tilde{\mu}}}(\nu')+\theta W_{A_{\tilde{\mu}}}(\nu'')\big)\neq 3^n$, that is, Condition (1) of this theorem holds.

Now we verify Condition (2) of Theorem \ref{thminimality}. For any $\mu,\mu'\in\mathbb{F}_{3}^{m*}$ with $\mu\neq\pm\mu'$, and any $\nu,\nu'\in\mathbb{F}_3^n$,
let $\Gamma=Re\big(W_F(\mu,\nu)+W_F(\mu',\nu')+W_F(\mu+\mu',\nu+\nu')-2W_F(\mu-\mu',\nu-\nu')\big)$.  Then by the relations of $\mu=(\mu_1,\tilde{\mu})$ and $\mu'=(\mu_1',\tilde{\mu}')$ with $\mu,\mu',\mu+\mu',\mu-\mu'\in\mathbb{F}_3^{m*}$, the discussions can be divided into 33 cases, which are given in Table \ref{b5}. Since the proof of each case is similar, we only give the detail proofs of the cases 1, 2 and 5.

{\bf Case 1.} $\mu_1=\mu_1'=0$, $\tilde{\mu}, \tilde{\mu}'$ and $\tilde{\mu}\pm\tilde{\mu}'$ are not $\mathbf{0_{m-1}}$. In this case, we obtain from Proposition \ref{propwalshtransform} that
\begin{align*}
\Gamma=Re\big(W_G(\tilde{\mu},\nu)+W_G(\tilde{\mu}',\nu')+W_G(\tilde{\mu}+\tilde{\mu}',\nu+\nu')
-2W_G(\tilde{\mu}-\tilde{\mu}',\nu-\nu')\big),
\end{align*}
which is not equal to $3^n$ for any $\nu,\nu'\in\mathbb{F}_3^n$ by Theorem \ref{thminimality}, since $C_G$ is minimal.

{\bf Case 2.} $\mu_1=0$, $ \tilde{\mu}\neq \mathbf{0_{m-1}}$, $\mu_1'=1$ and $\tilde{\mu}'=\mathbf{0_{m-1}}$. In this case, by Proposition \ref{propwalshtransform}, we have
\begin{align*}
\Gamma=Re\big(W_G(\tilde{\mu},\nu)+W_f(\nu')+W_{A_{\tilde{\mu}}}(\nu+\nu')
-2W_{A_{2\tilde{\mu}}}(\nu'-\nu)\big).
\end{align*}
Let $a=\nu'$ and $a'=\nu$. Then $a+a'=\nu+\nu'$, $a-a'=\nu'-\nu$, and
\begin{align*}
\Gamma=Re\big(W_f(a)+W_G(\tilde{\mu},a')+W_{A_{\tilde{\mu}}}(a+a')
-2W_{A_{2\tilde{\mu}}}(a-a')\big).
\end{align*}
Hence in this case, $\Gamma\neq 3^n$ for any $\nu,\nu'\in\mathbb{F}_3^n$ if and only if $Re\big(W_f(\nu)+W_G(\tilde{\mu},\nu')+W_{A_{\tilde{\mu}}}(\nu+\nu')
-2W_{A_{2\tilde{\mu}}}(\nu-\nu')\big)\neq 3^n$ for any $\nu,\nu'\in\mathbb{F}_3^n$. This corresponds to Condition (2) of this theorem with $\lambda_1=\lambda_2=\lambda_3=1$ and $\lambda_4=-2$.

{\bf Case 5.} $\mu_1=0$, $ \tilde{\mu}\neq \mathbf{0_{m-1}}$, $\mu_1'=1$, $\tilde{\mu}'\neq\mathbf{0_{m-1}}$ and $\tilde{\mu}\pm\tilde{\mu}'\neq\mathbf{0_{m-1}}$. In this case, by Proposition \ref{propwalshtransform}, we have
\begin{align*}
\Gamma=Re\big(W_G(\tilde{\mu},\nu)+W_{A_{\tilde{\mu}'}}(\nu')+W_{A_{\tilde{\mu}+\tilde{\mu}'}}(\nu+\nu')
-2W_{A_{2(\tilde{\mu}-\tilde{\mu}')}}(\nu'-\nu)\big).
\end{align*}
Let $a=\nu'$, $a'=\nu$. Then $a+a'=\nu+\nu'$ and $a-a'=\nu'-\nu$. Let $b=\tilde{\mu}', b'=\tilde{\mu}$. Then $b+b'=\tilde{\mu}+\tilde{\mu}'$ and $b-b'=\tilde{\mu}'-\tilde{\mu}=2(\tilde{\mu}-\tilde{\mu}')$. Hence, we have
\begin{align*}
\Gamma=Re\big(W_{A_{b}}(a)+W_G(b',a')+W_{A_{b+b'}}(a+a')
-2W_{A_{b-b'}}(a-a')\big).
\end{align*}
Therefore, in this case $\Gamma\neq 3^n$ for any $\nu,\nu'\in\mathbb{F}_3^n$ if and only if $Re\big(W_{A_{\tilde{\mu}}}(\nu)+W_G(\tilde{\mu}',\nu')+W_{A_{\tilde{\mu}+\tilde{\mu}'}}(\nu+\nu')
-2W_{A_{\tilde{\mu}-\tilde{\mu}'}}(\nu-\nu')\big)\neq 3^n$ for any $\nu,\nu'\in\mathbb{F}_3^n$. This corresponds to Condition (3) of this theorem with $\lambda_1=\lambda_2=\lambda_3=1$ and $\lambda_4=-2$.

The other cases and the corresponding values of $\lambda_1,\lambda_2,\lambda_3,\lambda_4$ are given in Table \ref{b5}, where Condition (3) of this theorem is determined by the cases 5, 9, 15, 18, 21, 27,  30 and 33 in Table \ref{b5}; and Condition (2) of this theorem is determined by the other cases in Table \ref{b5}.

The proof is then obtained from the above discussions.
\end{proof}

\begin{table*}[!htbp]\setlength{\abovecaptionskip}{0cm}
\caption{The cases of $\mu$ and $\mu'$ in the proof of Theorem \ref{thNS}}  \centering \label{b5}
\medskip
\begin{tabular}{|c|c|c|c|c|c|c|c|c|c|c|}
 \hline
Cases & $\mu_1$ & $\tilde{\mu}$ & $\mu_1'$ & $\tilde{\mu}'$ &$\tilde{\mu} +\tilde{\mu}'$ &
 $\tilde{\mu} -\tilde{\mu}'$ & $\lambda_1$ & $\lambda_2$ & $\lambda_3$ & $\lambda_4$ \\
 \hline
1 & $0 $ & $\neq\mathbf{0_{m-1}}$ & $0 $ & $\neq\mathbf{0_{m-1}}$ &$\neq\mathbf{0_{m-1}}$
& $\neq\mathbf{0_{m-1}}$ & - &-&-&- \\
 \hline
 2 & $0 $ & $\neq\mathbf{0_{m-1}}$ & $1 $ & $\mathbf{0_{m-1}}$ &$\neq\mathbf{0_{m-1}}$ &
 $\neq\mathbf{0_{m-1}}$ & 1 &1&1&-2 \\
 \hline
 3 & $0 $ & $\neq\mathbf{0_{m-1}}$ & $1 $ & $\neq\mathbf{0_{m-1}}$ &$\mathbf{0_{m-1}}$
 & $\neq\mathbf{0_{m-1}}$ & 1& 1 & -2 & 1 \\
 \hline
 4 & $0 $ & $\neq\mathbf{0_{m-1}}$ & $1 $ & $\neq\mathbf{0_{m-1}}$ &$\neq\mathbf{0_{m-1}}$
 & $\mathbf{0_{m-1}}$ & -2& 1 & 1 & 1 \\
 \hline
  5 & $0 $ & $\neq\mathbf{0_{m-1}}$ & $1 $ & $\neq\mathbf{0_{m-1}}$ & $\neq\mathbf{0_{m-1}}$
 & $\neq\mathbf{0_{m-1}}$ & 1& 1 & 1 & -2 \\
  \hline
   6 & $0 $ & $\neq\mathbf{0_{m-1}}$ & $-1 $ & $\mathbf{0_{m-1}}$ &$\neq\mathbf{0_{m-1}}$
 & $\neq\mathbf{0_{m-1}}$ & 1& 1 & -2 & 1 \\
 \hline
  7 & $0 $ & $\neq\mathbf{0_{m-1}}$ & $-1 $ & $\neq\mathbf{0_{m-1}}$ &$\mathbf{0_{m-1}}$
 & $\neq\mathbf{0_{m-1}}$ & 1& 1 & 1 & -2 \\
 \hline
  8 & $0 $ & $\neq\mathbf{0_{m-1}}$ & $-1 $ & $\neq\mathbf{0_{m-1}}$ &$\neq\mathbf{0_{m-1}}$
 & $\mathbf{0_{m-1}}$ & -2& 1 & 1 & 1\\
 \hline
  9 & $0 $ & $\neq\mathbf{0_{m-1}}$ & $-1 $ & $\neq\mathbf{0_{m-1}}$ &$\neq\mathbf{0_{m-1}}$
 & $\neq\mathbf{0_{m-1}}$ & 1& 1 & -2 & 1\\
 \hline
 10 & $1 $ & $ \mathbf{0_{m-1}}$ & $0 $ & $\neq\mathbf{0_{m-1}}$ &$\neq\mathbf{0_{m-1}}$
 & $\neq\mathbf{0_{m-1}}$ & 1& 1 & 1 & -2\\
  \hline
 11 & $1 $ & $ \mathbf{0_{m-1}}$ & $1 $ & $\neq\mathbf{0_{m-1}}$ &$\neq\mathbf{0_{m-1}}$
 & $\neq\mathbf{0_{m-1}}$ & 1& -2 & 1 & 1\\
 \hline
 12 & $1 $ & $ \mathbf{0_{m-1}}$ & $-1 $ & $\neq\mathbf{0_{m-1}}$ &$\neq\mathbf{0_{m-1}}$
 & $\neq\mathbf{0_{m-1}}$ & 1& 1 & -2 & 1\\
 \hline
 13 & $1 $ & $ \neq\mathbf{0_{m-1}}$ & $0 $ & $\neq\mathbf{0_{m-1}}$ &$\mathbf{0_{m-1}}$
 & $\neq\mathbf{0_{m-1}}$ & 1& 1 & -2 & 1\\
 \hline
 14 & $1 $ & $ \neq\mathbf{0_{m-1}}$ & $0 $ & $\neq\mathbf{0_{m-1}}$ &$\neq\mathbf{0_{m-1}}$
 & $\mathbf{0_{m-1}}$ & -2& 1 & 1 & 1\\
 \hline
 15 & $1 $ & $ \neq\mathbf{0_{m-1}}$ & $0 $ & $\neq\mathbf{0_{m-1}}$ &$\neq\mathbf{0_{m-1}}$
 & $\neq\mathbf{0_{m-1}}$ & 1& 1 & 1 & -2\\
 \hline
 16 & $1 $ & $ \neq\mathbf{0_{m-1}}$ & $1 $ & $\mathbf{0_{m-1}}$ &$\neq\mathbf{0_{m-1}}$
 & $\neq\mathbf{0_{m-1}}$ & 1& -2 & 1 & 1\\
 \hline
  17 & $1 $ & $ \neq\mathbf{0_{m-1}}$ & $1 $ & $\neq\mathbf{0_{m-1}}$ &$\mathbf{0_{m-1}}$
 &$\neq \mathbf{0_{m-1}}$ & 1& -2 & 1 & 1\\
  \hline
  18 & $1 $ & $ \neq\mathbf{0_{m-1}}$ & $1 $ & $\neq\mathbf{0_{m-1}}$ &$\neq\mathbf{0_{m-1}}$
 &$\neq \mathbf{0_{m-1}}$ & 1& -2 & 1 & 1\\
 \hline
 19 & $1 $ & $ \neq\mathbf{0_{m-1}}$ & $-1 $ & $\mathbf{0_{m-1}}$ &$\neq\mathbf{0_{m-1}}$
 & $\neq\mathbf{0_{m-1}}$ & 1& 1 & 1 & -2\\
 \hline
  20 & $1 $ & $ \neq\mathbf{0_{m-1}}$ & $-1 $ & $\neq\mathbf{0_{m-1}}$ &$\neq\mathbf{0_{m-1}}$
 & $\mathbf{0_{m-1}}$ & -2& 1 & 1 & 1\\
 \hline
 21 & $1 $ & $ \neq\mathbf{0_{m-1}}$ & $-1 $ & $\neq\mathbf{0_{m-1}}$ &$\neq\mathbf{0_{m-1}}$
 &  $\neq\mathbf{0_{m-1}}$ & -2& 1 & 1 & 1\\
 \hline
 22 & $-1 $ & $ \mathbf{0_{m-1}}$ & $0 $ & $\neq\mathbf{0_{m-1}}$ &$\neq\mathbf{0_{m-1}}$
 &  $\neq\mathbf{0_{m-1}}$ & 1& 1 & -2 & 1\\
 \hline
 23 & $-1 $ & $ \mathbf{0_{m-1}}$ & $1$ & $\neq\mathbf{0_{m-1}}$ &$\neq\mathbf{0_{m-1}}$
 &  $\neq\mathbf{0_{m-1}}$ & 1& 1 & 1 & -2\\
 \hline
 24 & $-1 $ & $ \mathbf{0_{m-1}}$ & $-1$ & $\neq\mathbf{0_{m-1}}$ &$\neq\mathbf{0_{m-1}}$
 &  $\neq\mathbf{0_{m-1}}$ & 1& -2 & 1 & 1\\
 \hline
 25 & $-1 $ & $ \neq\mathbf{0_{m-1}}$ & $0$ & $\neq\mathbf{0_{m-1}}$ &$\mathbf{0_{m-1}}$
 &  $\neq\mathbf{0_{m-1}}$ & 1& 1 & 1 & -2\\
 \hline
 26 & $-1 $ & $ \neq\mathbf{0_{m-1}}$ & $0$ & $\neq\mathbf{0_{m-1}}$ &$\neq\mathbf{0_{m-1}}$
 &  $\mathbf{0_{m-1}}$ & -2& 1 & 1 & 1\\
 \hline
  27 & $-1 $ & $ \neq\mathbf{0_{m-1}}$ & $0$ & $\neq\mathbf{0_{m-1}}$ &$\neq\mathbf{0_{m-1}}$
 &  $\neq\mathbf{0_{m-1}}$ & 1& 1 & -2 & 1\\
 \hline
 28 & $-1 $ & $ \neq\mathbf{0_{m-1}}$ & $1$ & $\mathbf{0_{m-1}}$ &$\neq\mathbf{0_{m-1}}$
 &  $\neq\mathbf{0_{m-1}}$ & 1& 1 & -2 & 1\\
 \hline
 29 & $-1 $ & $ \neq\mathbf{0_{m-1}}$ & $1$ & $\neq\mathbf{0_{m-1}}$ &$\neq\mathbf{0_{m-1}}$
 &  $\mathbf{0_{m-1}}$ & -2& 1 & 1 & 1\\
 \hline
 30 & $-1 $ & $ \neq\mathbf{0_{m-1}}$ & $1$ & $\neq\mathbf{0_{m-1}}$ &$\neq\mathbf{0_{m-1}}$
 &  $\neq\mathbf{0_{m-1}}$ & -2& 1 & 1 & 1\\
 \hline
 31 & $-1 $ & $ \neq\mathbf{0_{m-1}}$ & $-1$ & $\mathbf{0_{m-1}}$ &$\neq\mathbf{0_{m-1}}$
 &  $\neq\mathbf{0_{m-1}}$ & 1& -2 & 1 & 1\\
 \hline
 32 & $-1 $ & $ \neq\mathbf{0_{m-1}}$ & $-1$ & $\neq\mathbf{0_{m-1}}$ &$\mathbf{0_{m-1}}$
 &  $\neq\mathbf{0_{m-1}}$ & 1& -2 & 1 & 1\\
 \hline
 33 & $-1 $ & $ \neq\mathbf{0_{m-1}}$ & $-1$ & $\neq\mathbf{0_{m-1}}$ &$\neq\mathbf{0_{m-1}}$
 &  $\neq\mathbf{0_{m-1}}$ & 1& -2 & 1 & 1\\
 \hline
\end{tabular}
\end{table*}

Observe that though there are some symmetries in the conditions of Theorem \ref{thNS}, its proof is tedious (it discusses 33 cases). However, these discussions seem to be unavoidable, because this theorem is in fact a refined version of Theorem \ref{thminimality} (refine $F$ to be $F=(f,G)$), and so one needs to verify Conditions (1) and (2) of Theorem \ref{thminimality}, which naturally  leads to the 33 discussions by Proposition \ref{propwalshtransform}. Therefore, Theorem \ref{thNS} seems difficult to be extended to the case of $p\neq3$.

\begin{remark}\label{rmkm=2}
Note that when $m=2$, Condition (3) of Theorem \ref{thNS} does not exist, and hence all the three conditions of Theorem \ref{thNS} are reduced to the first two of them.
\end{remark}

\begin{remark}
It is  very challenging to find a minimal ternary linear code violating the AB condition from Construction \ref{con2}, because for that purpose, one has to find a function $f:\mathbb{F}_3^n\to\mathbb{F}_3$ and a vectorial function $G:\mathbb{F}_3^{n}\to \mathbb{F}_3^{m-1}$ such that the conditions of Construction \ref{con2} are satisfied (those conditions are difficult to satisfy, especially for Conditions ${\rm({\bf a})}$ and  ${\rm({\bf b})}$ since they imply that the linear code $C_f$ in \eqref{eqcodeBooleanintro} is minimal and violates the AB condition). Moreover, 
 one has also  to manage $f$ and $G$ such that  the three conditions of Theorem \ref{thNS} are satisfied.
\end{remark}

According to the fact $Re(z_1+z_2)\leq |z_1|+|z_2|$ for any two complex numbers $z_1,z_2$, Theorem \ref{thNS} can be simplified to the following corollary.

\begin{corollary}\label{corMAB}
The ternary linear code $C_F$ generated by Construction \ref{con2} is minimal if the following condition is satisfied:
\begin{align}\label{eqcorMAB}
\max_{\tilde{\mu}\in\mathbb{F}_3^{m-1}\backslash\{{\bf0_{m-1}}\},\nu\in\mathbb{F}_3^n}\bigg\{|W_f(\nu)|, |W_G(\tilde{\mu},\nu)|,|W_{A_{\tilde{\mu}}}(\nu)|\bigg\}<\frac{3^n}{5}.
\end{align}
\end{corollary}

\section{Several minimal ternary linear codes violating the AB condition}\label{sec:appli}

In this section, we present several minimal ternary linear codes violating the AB condition according to Proposition \ref{propviolateAB} and Theorem \ref{thNS}. To this end, we need first to find a function $f$ from $\mathbb{F}_3^n$ to $\mathbb{F}_3$ such that it satisfies the conditions of Construction \ref{con2}. Note that there is a known such function in the literature (see the function (9) and Theorem 18 in \cite{Heng-2018-FFA}). But for that function, we cannot find a vectorial function $G$ from $\mathbb{F}_3^n$ to $\mathbb{F}_3^{m-1}$ such that  the conditions of Theorem \ref{thNS} are satisfied. So we give another $f$ by using the following result.

\begin{lemma}\label{lemwalsh3fun}
Let $\phi_1,\phi_2,\phi_3$ be three functions from $\mathbb{F}_3^n$ to $\mathbb{F}_3$. Let
\begin{align}\label{3fun}
f(x)=\phi_1(x)+(\phi_2-\phi_1)(x)(\phi_3-\phi_1)(x).
 \end{align}
 Then the Walsh transform of $f$ at $\nu\in\mathbb{F}_3^n$ satisfies that
 \begin{align*}
 W_f(\nu)=&\frac{1}{3}\bigg(W_{\phi_1}(\nu)+W_{\phi_2}(\nu)+W_{2(\phi_1+\phi_2)}(\nu)+W_{\phi_3}(\nu)
 +W_{2(\phi_1+\phi_3)}(\nu)\\
 &~~~~+\zeta_3\big(W_{\phi_1+2\phi_2+\phi_3}(\nu)
 +W_{\phi_1+\phi_2+2\phi_3}(\nu)\big)+\zeta_3^2\big(W_{2\phi_1+\phi_2+\phi_3}(\nu)+W_{2(\phi_2+\phi_3)}(\nu)\big)\bigg).
 \end{align*}
\end{lemma}

\begin{proof}
Let $T_j:=\{x\in\mathbb{F}_3^n: \phi_2(x)-\phi_1(x)=j\}$ for $j\in\mathbb{F}_3$. Then the Walsh transform of $f$ satisfies that
\begin{align*}
W_{f}(\nu)=\sum_{x\in T_0}\zeta_3^{\phi_1(x)-\nu\cdot x}+\sum_{x\in T_1}\zeta_3^{\phi_3(x)-\nu\cdot x}+\sum_{x\in T_2}\zeta_3^{2\phi_1(x)+2\phi_3(x)-\nu\cdot x}.
\end{align*}
Note that for any $j\in\mathbb{F}_3$ and any function $h:\mathbb{F}_3^n\to \mathbb{C}$, it holds that
\begin{align*}
\sum_{x\in T_j}h(x)=\sum_{x\in\mathbb{F}_3^n}\frac{1+\zeta_3^{3-j}\zeta_3^{(\phi_2-\phi_1)(x)}
+\zeta_3^{3-2j}\zeta_3^{2(\phi_2-\phi_1)(x)}}{3}h(x).
\end{align*}
We have
\begin{align*}
\sum_{x\in T_0}\zeta_3^{\phi_1(x)-\nu\cdot x}
=&\frac{1}{3}\sum_{x\in\mathbb{F}_3^n}\bigg(\zeta_3^{\phi_1(x)-\nu\cdot x}+\zeta_3^{\phi_2(x)-\nu\cdot x}
+\zeta_3^{2(\phi_1+\phi_2)(x)-\nu\cdot x}\bigg)\\
=&\frac{1}{3}\big(W_{\phi_1}(\nu)+W_{\phi_2}(\nu)+W_{2(\phi_1+\phi_2)}(\nu)\big).
\end{align*}
Similarly, one has
\begin{align*}
\sum_{x\in T_1}\zeta_3^{\phi_3(x)-\nu\cdot x}=&\frac{1}{3}\big(W_{\phi_3}(\nu)+\zeta_3^2W_{2\phi_1+\phi_2+\phi_3}(\nu)+\zeta_3 W_{\phi_1+2\phi_2+\phi_3}(\nu)\big)\hspace{0.3 cm}{\rm and}\\
\sum_{x\in T_2}\zeta_3^{2\phi_1(x)+2\phi_3(x)-\nu\cdot x}=&\frac{1}{3}\big(W_{2(\phi_1+\phi_3)}(\nu)+\zeta_3W_{\phi_1+\phi_2+2\phi_3}(\nu)+\zeta_3^2 W_{2\phi_2+2\phi_3}(\nu)\big).
\end{align*}
The result follows then from the above discussions.
\end{proof}

For a generalized version of Lemma \ref{lemwalsh3fun} on $\mathbb{F}_2^n$, the reader is referred to \cite{Li et al.-2021}.

\begin{lemma}\label{lemvioteAB1}
 Let $n$ and $r$ be two integers with $n-r>2$. Let $a,b$ be two independent vectors of $\mathbb{F}_3^n$. Let $E$ be an $r$-dimensional linear subspace of $\mathbb{F}_3^n$ and $E^{\bot}$ denote the dual space of $E$, that is, $E^{\bot}=\{x\in\mathbb{F}_3^n: x\cdot y=0,~\forall~y\in E\}$. Let $1_{E}(x)$ be the indicator of $E$, that is, $1_{E}(x)$ equals $1$ if $x\in E$, and equals $0$ otherwise. Then the function $f$ from $\mathbb{F}_3^n$ to $\mathbb{F}_3$ defined by
\begin{align*}
 f(x)=1_{E}(x)+(a\cdot x)(b\cdot x)+2
 \end{align*}
 satisfies the conditions of Construction \ref{con2} if $a,b,a+b, a-b\notin E^{\bot}$.
\end{lemma}

\begin{proof}
Firstly, it is easily seen that $f(\mathbf{0_n})=0$. Now
let $\phi_1(x)=1_{E}(x)+2$, $\phi_2(x)=\phi_1(x)+a\cdot x$ and $\phi_3(x)=\phi_1(x)+b\cdot x$. We have $f(x)=\phi_1(x)+(\phi_2-\phi_1)(x)(\phi_3-\phi_1)(x)$. Note that $W_{\phi(x)+a\cdot x}(\nu)=W_{\phi}(\nu-a)$ for any  function $\phi:\mathbb{F}_3^n\to\mathbb{F}_3$ and any $a, \nu\in\mathbb{F}_3^n$. Then by Lemma \ref{lemwalsh3fun}, we obtain that
\begin{align}\label{eqwalshf+ab}
W_f(\nu)=&\frac{1}{3}\!\bigg(\!W_{\phi_1}(\nu)\!+\!W_{\phi_1}(\nu\!-\!a)\!+\!W_{\phi_1}(\nu\!-\!2a)
\!+\!W_{\phi_1}(\nu\!-\!b)\!+\!W_{\phi_1}(\nu\!-\!2b)\nonumber\\
&~~+\!\zeta_3\big(W_{\phi_1}(\nu\!-\!a\!-\!2b)\!+\!W_{\phi_1}(\nu\!-\!2a\!-\!b)\big)
\!+\!\zeta_3^2\big(W_{\phi_1}(\nu\!-\!a\!-\!b)\!+\!W_{\phi_1}(\nu\!-\!2a\!-\!2b)\big)\!\bigg)\!.
\end{align}
In addition, the Walsh transform of $\phi_1$ at $\nu\in\mathbb{F}_3^n$ is that
\begin{align*}
W_{\phi_1}(\nu)=\zeta_3^2\sum_{x\in\mathbb{F}_3^n}\zeta_3^{1_{E}(x)-\nu\cdot x}
=\zeta_3^2\big(\zeta_3+\sum_{x\in\mathbb{F}_3^{n*}}\zeta_3^{\nu\cdot x}+(\zeta_3-1)\sum_{x\in E^*}\zeta_3^{\nu\cdot x}\big).
\end{align*}
Recall that $E$ is an $r$-dimensional linear subspace of $\mathbb{F}_3^n$. We have  $\sum_{x\in E^*}\zeta_3^{\nu\cdot x}$ equals $\sum_{x\in E}\zeta_3^{\nu\cdot x}-1=-1$ if $\nu\notin E^{\bot}$, and equals $3^r-1$ otherwise. Thus, we obtain that
 \begin{align}\label{eqwalphi1}
W_{\phi_1}(\nu)=\begin{cases}
3^n\zeta_3^2+3^r-3^r\zeta_3^2, &{\rm if}~\nu=\mathbf{0_n},\\
3^r-3^r\zeta_3^2, &{\rm if}~\nu\in E^{\bot}\backslash\{\mathbf{0_n}\},\\
0,&{\rm otherwise.}
\end{cases}
\end{align}
Recall that $a,b\in\mathbb{F}_3^n$ are independent. Now if $a,b,a+b,a-b\notin E^{\bot}$, then from \eqref{eqwalshf+ab} and \eqref{eqwalphi1}, we obtain that
\begin{align*}
W_f(\nu)=\begin{cases}
\frac{1}{3}(3^{n}\zeta_3^2+3^{r}-3^{r}\zeta_3^2), &{\rm if}~\nu\in\{\mathbf{0_n}, a,2a, b, 2b\},\\
\frac{1}{3}\zeta_3(3^{n}\zeta_3^2+3^{r}-3^{r}\zeta_3^2), &{\rm if}~\nu\in\{a+2b, 2a+b\},\\
\frac{1}{3}\zeta_3^2(3^{n}\zeta_3^2+3^{r}-3^{r}\zeta_3^2), &{\rm if}~\nu\in\{a+b, 2a+2b\},\\
\frac{1}{3}(3^{r}-3^{r}\zeta_3^2), &{\rm if}~\nu\in
(E^{\bot}\cup a+E^{\bot}\cup 2a+E^{\bot}\cup b+E^{\bot} \cup 2b+E^{\bot})\backslash T,\\
\frac{1}{3}\zeta_3(3^{r}-3^{r}\zeta_3^2), &{\rm if}~\nu\in (a+2 b+E^{\bot}\cup 2a+ b+E^{\bot})\backslash T,\\
\frac{1}{3}\zeta_3^2(3^{r}-3^{r}\zeta_3^2), &{\rm if}~\nu\in (a+b+E^{\bot}\cup 2a+ 2b+E^{\bot})\backslash T,\\
0, &{\rm otherwise},
\end{cases}
\end{align*}
where $T$ is a 2-dimensional linear space of $\mathbb{F}_3^n$ spanned by $a$ and $b$, from which we derive that
\begin{align}\label{eqrewfab}
Re(W_f(\nu))=\begin{cases}
-\frac{3^{n-1}}{2}+\frac{3^{r}}{2}, &{\rm if}~\nu\in\{\mathbf{0_n}, a,2a, b, 2b\},\\
3^{n-1}-\frac{3^{r}}{2}, &{\rm if}~\nu\in\{a+2b, 2a+b\},\\
-\frac{3^{n-1}}{2}, &{\rm if}~\nu\in\{a+b, 2a+2b\},\\
\frac{3^{r}}{2}, &{\rm if}~\nu\in
(E^{\bot}\cup a+E^{\bot}\cup 2a+E^{\bot}\cup b+E^{\bot} \cup 2b+E^{\bot})\backslash T,\\
-\frac{3^{r}}{2}, &{\rm if}~\nu\in (a+2 b+E^{\bot}\cup 2a+ b+E^{\bot})\backslash T,\\
0, &{\rm otherwise}.
\end{cases}
\end{align}
Then we have
\begin{align*}
3\max_{\nu\in\mathbb{F}_3^n}Re(W_f(\nu))-2\min_{\nu\in\mathbb{F}_3^n}Re(W_f(\nu))
=3(3^{n-1}-\frac{3^r}{2})-2(-\frac{3^{n-1}}{2})=3^n+3^{n-1}-\frac{3^{r+1}}{2},
\end{align*}
which is greater than $3^n$ when $n-r> 2$. Thus, Condition $({\bf b})$ of Construction \ref{con2} is satisfied.
Moreover, according to \eqref{eqrewfab}, one can also verify that
\begin{align*}
Re(W_f(\nu)+W_f(\nu')+\theta W_f(\nu''))\neq 3^n
\end{align*}
for any $\nu,\nu',\nu''\in\mathbb{F}_3^n$ and any $\theta\in\{1,-2\}$. Thus, Condition $({\bf a})$ of Construction \ref{con2} is also satisfied. This completes the proof.
\end{proof}

In our proof, we also need the following result.

\begin{lemma}\label{lemminimalplate}
Let $g$ be a regular $s$-plateaued function from $\mathbb{F}_3^n$ to $\mathbb{F}_3$ such that $g(\mathbf{0_n})=0$. Let $\phi(x)=g(x)+1_{E}(x)+2$, where $E$ is an $r$-dimensional linear subspace of  $\mathbb{F}_3^n$. Then  for any $\nu\in\mathbb{F}_3^n$, it holds that
\begin{align*}
\big|W_{\phi}(\nu)\big|\leq 3^{\frac{n+s}{2}}+2\times 3^r.
\end{align*}
\end{lemma}
\begin{proof}
For any $\nu\in\mathbb{F}_3^n$, it is easy to obtain that
\begin{align*}
W_{\phi}(\nu)=&\zeta_3^2\big(W_{g}(\nu)+(\zeta_3-1)\sum_{x\in E}\zeta_3^{g(x)-\nu\cdot x}\big).
\end{align*}
Since $g$ is regular $s$-plateaued, we have $\zeta_3^2W_g(\nu)\in \{0, 3^{\frac{n+s}{2}}\zeta_3^{2+\tilde{g}(\nu)}\}$ for some function $\tilde{g}:\mathbb{F}_3^n\to\mathbb{F}_3$ and for any $\nu\in\mathbb{F}_3^n$. Thus, $\big|\zeta_3^2W_g(\nu)\big|\leq 3^{\frac{n+s}{2}}$ for any $\nu\in\mathbb{F}_3^n$.
Since $E$ is an $r$-dimensional linear subspace of $\mathbb{F}_3^n$, we have $\big|(\zeta_3-1)\sum_{x\in E}\zeta_3^{g(x)-\nu\cdot x}\big)\big|\leq 2\times 3^{r}$ for any $\nu\in\mathbb{F}_3^n$. Combining them, we obtain the desired result.
\end{proof}

Now we give the main result of this section.

\begin{theorem}\label{thmvabpla}
 Let $n>6$ and $ m\geq 2$ be two integers. Let $E$ be an $r$-dimensional linear subspace of $\mathbb{F}_3^n$ with $n-r>3$, and let $f(x)=1_{E}(x)+(a\cdot x)(b\cdot x)+2$, where $a,b\in\mathbb{F}_3^n$ are independent with $a,b,a+b,a-b\notin E^{\bot}$.
 Let $G: \mathbb{F}_3^n\to \mathbb{F}_3^{m-1}$ be a vectorial regular $s$-plateaued function with $0\leq s\leq n-6$ and $G(\mathbf{0_n})=\mathbf{0_{m-1}}$.  Let $F=(f, G)$. Then the ternary linear code $C_F$ generated by Construction \ref{ourcon} is a $[3^n-1, n+m, 3^{n-1}+3^{n-2}+3^{r-1}]$ code violating the AB condition.
\end{theorem}

\begin{proof}
From Lemma \ref{lemvioteAB1} and Theorem \ref{th3weirgl}, we obtain that $f$ and $G$ satisfy all the conditions of Construction \ref{con2}. From Proposition \ref{propviolateAB}, we derive that $C_F$ violates the AB condition. Below, we show that $C_F$ is minimal by verifying the three conditions of Theorem \ref{thNS}.

For any $\mathbf{0_{m-1}}\neq\tilde{\mu}\in\mathbb{F}_3^{m-1}$, let $A_{\tilde{\mu}}(x):=f(x)+\tilde{\mu}\cdot G(x)=\phi_{\tilde{\mu}}(x)+(a\cdot x)(b\cdot x)$, where $\phi_{\tilde{\mu}}(x)=\tilde{\mu}\cdot G(x)+1_{E}(x)+2$. Since $G$ is a vectorial regular $s$-plateaued function, by  Lemma \ref{lemminimalplate}, we have
$$\big|W_{\phi_{\tilde{\mu}}}(\nu)\big|\leq 3^{\frac{n+s}{2}}+2\times 3^r, ~\forall~\tilde{\mu}\neq \mathbf{0_{m-1}},\nu\in\mathbb{F}_3^n.$$
By \eqref{eqwalshf+ab}, we have
\begin{align*}
\big|W_{A_{\tilde{\mu}}}(\nu)\big|\leq 3\max_{\nu\in\mathbb{F}_3^n}\big|W_{\phi_{\tilde{\mu}}}(\nu)\big|, ~\forall ~\tilde{\mu}\neq\mathbf{0_{m-1}}.
\end{align*}
Therefore, we obtain that
\begin{align}\label{eqreAmu}
-{3^{\frac{n+s}{2}+1}}-2\times 3^{r+1}\leq Re(W_{A_{\tilde{\mu}}}(\nu))
\leq 3^{\frac{n+s}{2}+1}+2\times 3^{r+1}, ~\forall ~\tilde{\mu}\neq\mathbf{0_{m-1}}, \nu\in\mathbb{F}_3^n,
\end{align}
from which one easily derives that
\begin{align*}
Re\big(W_{A_{\tilde{\mu}}}(\nu)\!+\!W_{A_{\tilde{\mu}}}(\nu')\!+\!\theta W_{A_{\tilde{\mu}}}(\nu'')\big)\!\leq\! 4(3^{\frac{n+s}{2}+1}\!+\!2\!\times\! 3^{r+1}), \forall ~\tilde{\mu}\neq\mathbf{0_{m-1}},\nu,\nu',\nu''\in\mathbb{F}_3^n,\theta\in\{1,-2\},
\end{align*}
which is strictly less than $3^n$ if $n-r>3$ and $0\leq s\leq n-6$, since in this case $4\times3^{\frac{n+s}{2}+1}\leq 12\times 3^{n-3}$ and $8\times 3^{r+1}\leq 8\times 3^{n-3}$. Thus, Condition (1) of Theorem \ref{thNS} is satisfied. Since $G$ is a vectorial regular $s$-plateaued function from $\mathbb{F}_3^n$ to $\mathbb{F}_3^{m-1}$, we have $\big|Re(W_{G}(\tilde{\mu}, \nu))\big|\leq 3^{\frac{n+s}{2}}$ for any $\tilde{\mu}\in\mathbb{F}_{3}^{m-1}\backslash\{\mathbf{0_{m-1}}\}$, $\nu\in\mathbb{F}_3^n$. In addition,
according to  \eqref{eqrewfab} and \eqref{eqreAmu}, we have
\begin{align*}
\big|Re(W_f(\nu))\big|\leq 3^{n-1}-\frac{3^r}{2}, \big|Re(W_{A_{\tilde{\mu}}}(\nu))\big|\leq 3^{\frac{n+s}{2}+1}+2\times 3^{r+1},~~\forall~\tilde{\mu}\neq\mathbf{0_{m-1}}, \nu\in\mathbb{F}_3^n.
\end{align*}
Therefore, for any $\lambda_1,\lambda_2,\lambda_3, \lambda_4\in\{1,-2\}$ with one being $-2$, any $\tilde{\mu}, \tilde{\mu}',\tilde{\mu}\pm \tilde{\mu}'\in\mathbb{F}_3^{m-1}\backslash\{\mathbf{0_{m-1}}\}$, and any $\nu,\nu'\in\mathbb{F}_3^n$, we obtain that
\begin{align*}
&Re\big(\lambda_1W_f(\nu)+\lambda_2W_{G}(\tilde{\mu},\nu')+\lambda_3W_{A_{\tilde{\mu}}}(\nu+\nu')
+\lambda_4W_{A_{2\tilde{\mu}}}(\nu-\nu')\big)\\
\leq &2(3^{n-1}-\frac{3^r}{2})+ 3^{\frac{n+s}{2}}
+2(3^{\frac{n+s}{2}+1}+2\times 3^{r+1})\\
=& 2\times 3^{n-1}+7\times 3^{\frac{n+s}{2}}+11\times3^r,
\end{align*}
and
\begin{align*}
&Re\big(\lambda_1W_{A_{\tilde{\mu}}}(\nu)+\lambda_2W_{G}(\tilde{\mu}',\nu')
+\lambda_3W_{A_{\tilde{\mu}+\tilde{\mu}'}}(\nu+\nu')+\lambda_4W_{A_{\tilde{\mu}-\tilde{\mu}'}}(\nu-\nu')\big)\\
\leq &3^{\frac{n+s}{2}}+4(3^{\frac{n+s}{2}+1}+2\times 3^{r+1})\\
=&13\times 3^{\frac{n+s}{2}}+24\times 3^r.
\end{align*}
These two relations are strictly less than $3^n$ when $n-r\geq 4$ and $0\leq s\leq n-6$. This shows that Conditions (2) and (3) of Theorem \ref{thNS} are also satisfied, and hence $C_F$ is a minimal ternary linear code violating the AB condition.

Finally, by Propositions \ref{propdim} and \ref{propwalshtransform}, Remark \ref{rmkdistance}, and Relations \eqref{eqrewfab} and \eqref{eqreAmu}, we obtain that $C_F$ is a $[3^n-1, n+m, 3^{n-1}+3^{n-2}+3^{r-1}]$ code. The proof is finished.
\end{proof}

Note that when $s=0$ in Theorem \ref{thmvabpla}, that is, $G$ is a vectorial regular bent function from $\mathbb{F}_3^n$ to $\mathbb{F}_3^{m-1}$, the value of $m$ should be less than or equal to $\frac{n}{2}+1$, since vectorial regular bent functions exist only for $m-1\leq \frac{n}{2}$.

\begin{remark}
In section VI of \cite{LYJetal-2023-IT}, the authors presented several infinite families of minimal binary linear codes violating the AB condition from some special vectorial Boolean functions $F:\mathbb{F}_2^n\to\mathbb{F}_2^m$, where $n$ is even and $2\leq m\leq \frac{n}{2}+1$ or $m=n+1$. While for the other cases (including the case $n$ odd and other values of $m$),  the authors left them as a open question in Remark 12 of \cite{LYJetal-2023-IT}. In this section, we also presented several minimal ternary linear codes violating the AB condition from some special vectorial functions $F:\mathbb{F}_3^n\to\mathbb{F}_3^m$. In addition, our result (Theorem \ref{thmvabpla})  does not have ``restrictions" to the values of $n$ and $m$. Therefore, similarly as this section, one can use our method to solve the left question in Remark 12 of \cite{LYJetal-2023-IT}.
\end{remark}

\begin{remark}
Theorem \ref{thNS} is a general result for finding new minimal ternary linear codes violating the AB condition. In this section, we have given an application of Theorem \ref{thNS} by finding special functions $f$ and $G$. It would be interesting if someone can give other applications of Theorem \ref{thNS}.
\end{remark}

\section{ Concluding remarks}\label{sec:conclusion}

The main {contributions of this paper are presented as follows}:

 \begin{itemize}
\item A necessary and sufficient condition for a large class of ternary linear codes to be minimal (Theorem \ref{thminimality}).

\item Several {minimal ternary linear codes with three-weight and complete} weight distribution (Theorem \ref{th3weirgl}, Corollaries \ref{cors=0} and \ref{cors=1}).

\item A generic construction of ternary linear codes violating the AB condition from vectorial functions (Construction \ref{con2}).

\item A necessary and sufficient condition for the ternary linear codes in Construction \ref{con2} to be minimal (Theorem \ref{thNS}).

\item Several minimal linear codes violating the AB condition (Theorem \ref{thmvabpla}).
\end{itemize}

We highlight that, to the best of our knowledge, this is the first paper to construct minimal ternary linear codes from vectorial functions $F:\mathbb{F}_3^n\to\mathbb{F}_3^m$. Moreover, the minimal ternary linear codes obtained from Theorem \ref{thmvabpla} are valid for any positive integers $n\geq 6$ and $m\geq 2$ (without any restriction for $n$ even and $2\leq m\leq \frac{n}{2}$ or $m=n+1$, while those restrictions are required in Section VI of \cite{LYJetal-2023-IT}).
Hence, by using our method given in Section \ref{sec:appli}, one can solve the question raised in Remark 12 of \cite{LYJetal-2023-IT}.

\section*{Acknowledgments}

This work was supported in part by the National Key Research and Development Program of China under Grant 2019YFB2101703; in part
by the National Natural Science Foundation of China under Grants  62302001, 62372221, 61972258, 62272107 and U19A2066, in part by the China Postdoctoral Science Foundation under Grant 2023M740714, in part by the
Innovation Action Plan of Shanghai Science and Technology under Grants  20511102200 and 21511102200, in part by the Key Research and Development Program of Guangdong Province  under Grant 2020B0101090001,
 in part by the Natural Science Foundation for the Higher Education Institutions of Anhui Province under Grant 2023AH050250.

\end{document}